\documentclass[lettersize,journal]{IEEEtran}
\usepackage{algorithm}
\usepackage{algpseudocode}
\usepackage{amsmath,amsfonts}
\usepackage{amsthm}
\usepackage{array}
\usepackage{booktabs}
\usepackage{balance}
\usepackage{cuted}
\usepackage{comment}
\usepackage{color}
\usepackage{caption}
\usepackage{dsfont}
\usepackage{enumerate}
\usepackage{float}  
\usepackage{graphicx}
\usepackage[colorlinks=true,linkcolor=red,citecolor=red]{hyperref}
\usepackage{makecell}
\usepackage[numbers,sort&compress]{natbib}
\usepackage{subfigure}
\usepackage{stfloats}
\usepackage{siunitx}
\usepackage{subcaption}
\usepackage{setspace}
\usepackage{textcomp}
\usepackage{url}
\usepackage{verbatim}
\usepackage{xcolor}

\graphicspath{{pictures/}}

\newtheorem{theorem}{Theorem}

\IEEEdisplaynontitleabstractindextext
\IEEEpeerreviewmaketitle

\begin{document}

\title{Directional WPT Charging for Routing-Asymmetric WRSNs with a Mobile Charger }
\author{
    Zhenguo~Gao*~\IEEEmembership{Senior Member,~IEEE}, Qi Zhang, Qingyu Gao, Yunlong Zhao, Hsiao-Chun Wu*~\IEEEmembership{Fellow,~IEEE}
    
    \IEEEcompsocitemizethanks{
        \IEEEcompsocthanksitem Z.~G.~Gao, Q.~Zhang, and Q.~Y.~Gao are with both the Department of Computer Science and Technology in Huaqiao University, and Key Laboratory of Computer Vision and Machine Learning(Huaqiao University), Fujian Province University, Xiamen, FJ, 361021, CHINA. (e-mail: {\tt gzg@hqu.edu.cn}); Y.~L.~Zhao is with the School of Computer Science and Technology, Nanjing University of Aeronautics and Astronautics, Nanjing, JS, 211100, CHINA; H.-C.~Wu is with the School of Electrical Engineering and Computer Science, Louisiana State University, Baton Rouge, LA 70803, USA and also with the Innovation Center for AI Applications, Yuan Ze University, Chungli 32003, Taiwan (e-mail: {\tt hwu1@lsu.edu}.
    }
    
    \thanks{This work was jointly supported by Natural Science Foundation of China under Grants 62372190, 62072236, 61972166 and Industry University Cooperation Project of Fujian Province under Grant 2021H6030.}
}

\markboth{xxx,~Vol.~18, No.~9, September~2024}%
{xxx}

\maketitle

\begin{abstract}
Mobile Charge Scheduling for wirelessly charging nodes in Wireless Rechargeable Sensor Networks (WRSNs) is a promising but still evolving research area. Existing research mostly assumes a symmetric environment, where the routing costs in opposite directions between two locations are considered identical. However, various factors such as terrain restrictions and wind or water flows may invalidate the routing-symmetric assumption in practical environments, thereby significantly limiting the performance of these solutions in routing-asymmetric WRSNs (RA-WRSNs). To address the routing-asymmetric challenges in mobile charge scheduling for WRSNs, this paper systematically investigates the underlying Asymmetric Directional Mobile Charger (DMC) Charge Scheduling (ADMCCS) problem, aiming to minimize energy loss while satisfying the charging demands of the network nodes. The DMC model is assumed because its results can be easily applied to the specialized case of an Omnidirectional Mobile Charger (OMC). To solve the ADMCCS problem, we propose a four-step framework. First, a minimum-size efficient charging position set is selected using our designed K-means-based Charging Position Generation (KCPG) algorithm, addressing the challenge of the unlimited charging position selection space. Next, minimum-size functional-equivalent direction sets at these positions are determined using an optimal algorithm, tackling the challenge of infinite charging directions. Subsequently, the optimal energy transmission time lengths for all directions at the positions are obtained by formulating and solving a Nonlinear Program (NLP) problem. Finally, the Lin-Kernighan Heuristic (LKH) algorithm for the Asymmetric Traveling Salesman Problem is adapted to obtain a highly probable optimal loop tour, addressing the routing-asymmetric challenge. The combination of these steps results in our DMC Scheduling algorithm for RA-WRSNs (RA-DMCS). The properties of the ADMCCS problem and the proposed algorithms are analyzed, and experiments demonstrate that RA-DMCS considerably outperforms other typical algorithms.
\end{abstract}


\section{Introduction}
\label{sec:intro}
\IEEEPARstart{W}ireless Sensor Networks (WSNs) have been widely applied in various industrial and everyday scenarios. However, a major limitation hindering the rapid development of WSNs is the energy constraints of battery-powered nodes~\cite{wang_secure_2022,ma_path_2019}. This issue is particularly pronounced with the advancement of edge intelligence, which imposes high computational demands on the nodes and leads to increased energy consumption~\cite{meng_learning-driven_2021,fu_optimal_2016}. To address this issue, researchers are working not only on improving battery storage capacity but also on advancing Wireless Power Transfer (WPT) technology for on-demand wireless charging of the nodes~\cite{fan_survey_2018,liguang_xie_wireless_2013}.

With the ongoing development of WPT technology, researchers have extended its application to WSNs, leading to the concept of Wireless Rechargeable Sensor Networks (WRSNs)~\cite{Gao2023iotj}. In WRSNs, the nodes are equipped with energy reception modules, while fixed or Mobile Chargers (MCs) equipped with WPT energy transmission modules are responsible for charging the nodes on demand. A WRSN also contains a Base Station (BS) where the MCs are replenished and remain stationed, awaiting node charging requests.

To efficiently and promptly meet the energy demands of the nodes in WRSNs, optimizing the MCs' charging schedule involves designing their charging trajectory, scheduling their charging directions, and determining charging times along each direction. Executing a charging schedule with an MC incurs two types of energy consumption: movement energy consumption and charging energy consumption. The former is used for moving the MC, while the latter results from the energy transmission from the MC to the nodes. An MC's charging schedule should aim to minimize overall energy consumption while ensuring that all nodes receive adequate energy~\cite{dorigo_ant_2006}. Depending on the energy transmission module they are equipped with, MCs can be categorized as Omnidirectional MCs (OMCs) or Directional MCs (DMCs). OMCs radiate energy signals uniformly in all directions, while DMCs concentrate signal energy in a specific direction within a sector-like region~\cite{george_review_2020}.

Initially constrained by limited charging distance, early work focused on a one-to-one (O2O) charging mode using OMCs, where an MC must approach very close to a node to charge it~\cite{Lin2018jss}. This simplifies the charging schedule problem to a trajectory design problem for the MC. With advancements in WPT enabling long-distance wireless charging over several meters, many studies have begun to explore one-to-many (O2M) and many-to-many (M2M) charging modes. These approaches allow multiple MCs to charge several nodes simultaneously. In these studies, the charging schedule involves both the trajectory design of the MCs and their energy transmission schedules. Some research further extends the MC charging scenario by assuming that nodes can also transmit energy, enabling explicit multi-hop WPT. These studies propose efficient scheduling algorithms that produce improved schedules~\cite{Gao2020tgcn,Gao2023iotj,Gao2023adhoc}. DMC-based charging of WRSNs offers the advantage of higher energy transfer efficiency~\cite{liu_objective-variable_2020}, but it introduces the challenge of determining an optimal charging direction set. As a result, DMC-based charging of WRSNs is increasingly attracting research attention~\cite{gao_scheduling_2023}.

Existing studies often assume routing-symmetric environments~\cite{wu_collaborated_2019}, where the energy consumption for movement between two positions is considered identical in both directions. However, this routing-symmetric assumption often does not hold in many real-world scenarios due to factors such as terrain and topography restrictions, altitude differences, wind forces (if deployed in the air), and water flow (if deployed on or in water). Routing asymmetry makes designs intended for routing-symmetric environments ineffective in practical routing-asymmetric scenarios. Although the asymmetry of up-down paths in underwater sensor networks was addressed in~\cite{Lin2018jss}, this work only considered routing asymmetry within a simple greedy trajectory design framework and focused on OMCs in O2O charging mode. Consequently, it did not explore the more energy-efficient M2M mode, resulting in less effective charging schedules. For simplicity, we will omit the term "routing" when referring to "symmetric" and "asymmetric" in the following text.

To systematically address routing-asymmetric issues, we conducted an in-depth investigation into Directional WPT Charging of WRSNs using DMCs in such asymmetric environments. To avoid being limited to specific asymmetric scenarios and to generalize across various environments, we refer to these WRSNs as Routing-Asymmetric WRSNs (RA-WRSNs).

In this paper, we primarily investigate the DMC charge scheduling problem in RA-WRSNs, which we term the Asymmetric DMC Charge Scheduling (ADMCCS) problem. The objective is to develop a DMC charging schedule that ensures all nodes' energy demands are met. Several challenges impede this task: the unlimited selection space for charging positions, the infinite number of possible charging directions, and the routing asymmetry of the WRSNs.

We show that the ADMCCS problem is NP-hard, making it impractical to design an optimal algorithm; hence, we focus on efficient approximation algorithms. To simplify the problem, we additionally assume that the DMC does not transmit energy signals while moving. Inspired by~\cite{gao_energy_2019}, we address the ADMCCS problem using a multi-step approach, as outlined in Fig.~\ref{fig1}. First, we identify several charging positions within the network region. Next, we establish the charging directions and determine the time duration for transmitting energy along each direction at each position. Finally, we design a round tour to visit all charging positions in the RA-WRSNs.

The corresponding tasks in these steps are formulated and solved with specifically designed algorithms. Specifically, we propose a K-means Charging Position Generation (KCPG) algorithm to determine charging positions, addressing the challenge of the unlimited charging position selection space. We adopt the cMFRDS algorithm from~\cite{gao_scheduling_2023} to identify a minimum set of functional representative directions as the final charging directions, thus addressing the challenge of infinite charging directions. We use CPlex software to solve a linear programming problem and determine the optimal charging times for all directions at the selected charging positions. We use the Lin-Kernighan Heuristic (LKH) algorithm~\cite{lkh_2019}, a state-of-the-art algorithm for solving the TSP problem, to find an optimal loop tour with minimal length in the RA-WRSNs to visit all charging positions, thereby tackling the routing asymmetry challenge. These algorithms are then integrated into our DMC Scheduling algorithm in RA-WRSNs (RA-DMCS) to solve the ADMCCS problem. Extensive simulations and test-bed experiments demonstrate the superiority of our RA-DMCS algorithm.

\begin{figure}[htb]        
  \centering
  \includegraphics[width=0.49\textwidth]{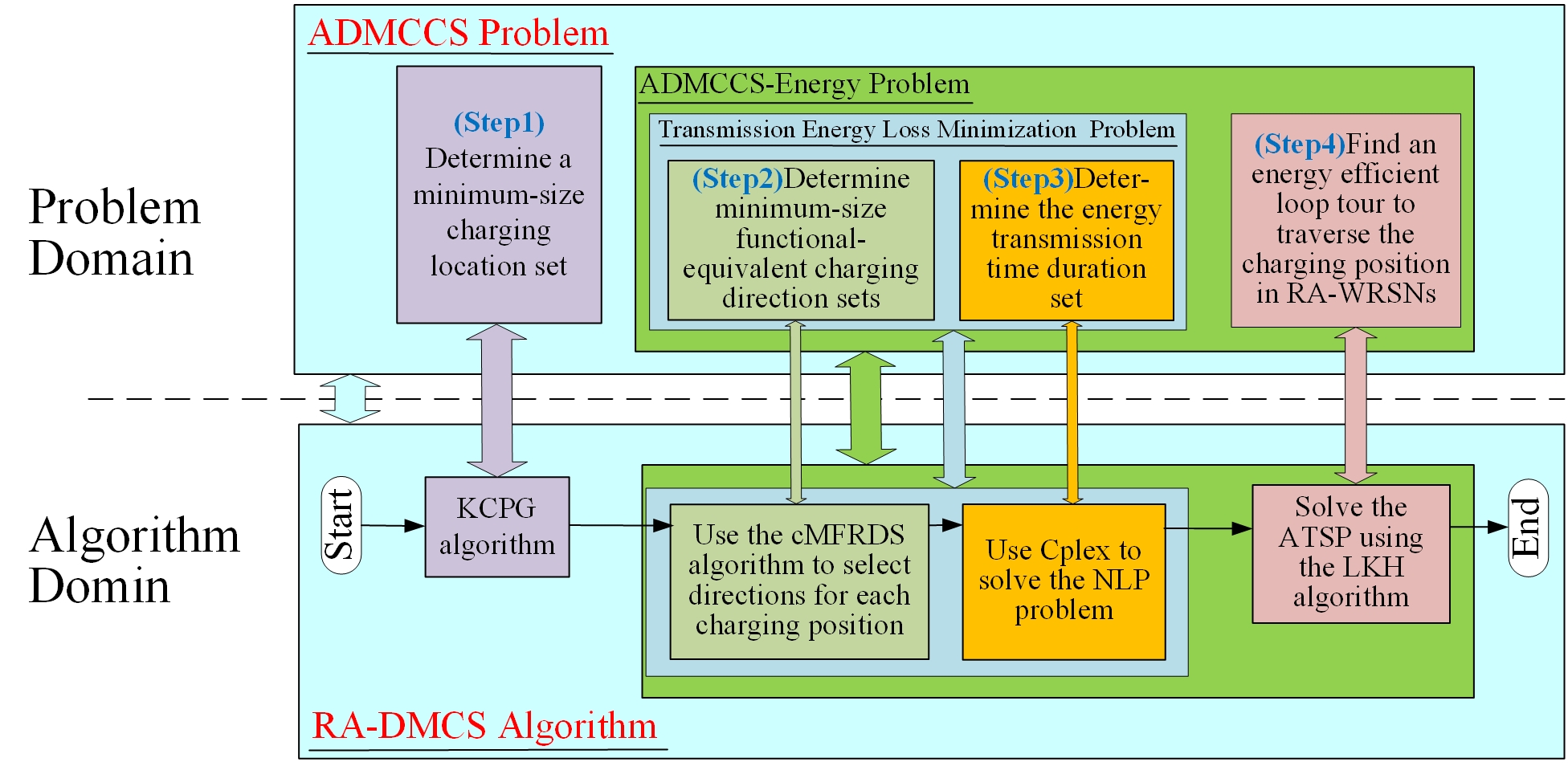}\\
  \caption{Outline of the RA-DMCS algorithm for the ADMCCS problem}
  \label{fig1}
\end{figure}

The main contributions of this paper are summarized as follows:

\begin{enumerate}
\item To the best of the authors' knowledge, this is the first work to abstractly construct the concept of routing-asymmetric WRSNs and to systematically investigate the DMC-based directional wireless charging scheduling problem in RA-WRSNs.

\item We demonstrate that the ADMCCS problem for DMC-based wireless charging scheduling in RA-WRSNs is NP-hard and apply state-of-the-art methods and results from the fields of the Traveling Salesman Problem (TSP) and Asymmetric TSP (ATSP) to address this problem.

\item We propose the RA-DMCS algorithm to solve the ADMCCS problem. Extensive simulation experiments validate the effectiveness of our RA-DMCS algorithm.
\end{enumerate}

The remaining sections of this paper are organized as follows: Section~\ref{sec:relate} provides a brief overview of previous research on WRSNs. Section~\ref{sec:model} introduces the system and energy transmission models. Section~\ref{sec:ADMCCS} details the ADMCCS problem, its mathematical formulation, and proposes corresponding solutions through in-depth analysis. Section~\ref{sec:RA-DMCS} presents the RA-DMCS algorithm along with its analysis. Section~\ref{sec:simula} evaluates the performance of the algorithms through simulations. Finally, Section~\ref{sec:conclusion} summarizes the key findings and concludes the paper.

\section{Related Work}
\label{sec:relate}
In this section, we summarize the current research on charging scheduling and classify it into two main categories: stationary charging and mobile charging. Additionally, we review related work on the study of asymmetric paths.

\subsection{Charging with Stationary Chargers}
As the deployment of stationary chargers with omnidirectional WPT can easily be adapted from various node deployment solutions for covering a network in closely related realms, most related work has focused on the topic of charging WRSNs with stationary directional chargers. The main concern is to optimally determine the positions and charging directions for the chargers while considering various additional factors.

For example, for a given WRSN and a fixed number of stationary chargers, Dai et al.~\cite{dai_optimizing_2017} proposed an approximate algorithm for optimizing the overall expected charging utility by jointly optimizing the positions and charging directions for the chargers. The approximation ratio of the algorithm was analyzed. For WRSNs where nodes may drift within a certain range, Wang et al.~\cite{wang_robust_2019} aimed to maximize the overall expected charging utility by determining the charging directions of multiple directional stationary chargers with predetermined fixed positions. The authors subsequently extended their work to complex scenarios with obstacles and multi-type heterogeneous stationary chargers, where a piecewise constant function was used to approximate the nonlinear energy transmission power model. This body of work targets 2D WRSNs, where the solution space is a restricted region in a 2D plane.

Jiang et al.~\cite{jiang_efficient_2016} proposed a 3D WRSN charging scheme in which nodes and chargers are deployed on different planes. By reasonably controlling the number and positions of chargers, the scheme aims to meet the charging requirements of the nodes while minimizing costs.

Although these solutions can produce efficient deployment strategies for the chargers, the lack of mobility limits their adaptability to changes in the energy demands of the nodes, thus restricting the flexibility of such WRSNs.

\subsection{Charging with Mobile Chargers}
For MCs with omnidirectional WPT, charging scheduling with mobile chargers involves trajectory design and charging time determination.

MCs with directional WPT introduce the selection of charging direction as a new design dimension, and the charging time at a position should be specified according to the charging directions.

For instance, Wu et al.~\cite{wu_charging_2019} proposed flexible scheduling strategies to optimize charging routes, enhancing the charging utility of the MCs. Liu et al.~\cite{liu_effective_2020} further refined the optimization objectives by focusing on energy utilization efficiency and the number of dead nodes. They proposed a request-based charging path system tailored to their model. Lin et al.~\cite{lin_maximizing_2020} improved the path planning component by exploring the selection of charging positions while considering energy constraints and obstacles, thereby offering solutions for complex network environments. In~\cite{dai_placing_2020}, Dai improved upon previous work by incorporating directional charging into their scheme and studying the placement strategy of DMCs under mobility constraints. This strategy ensures effective energy delivery even when movement is restricted. Additionally, Abhinav et al.~\cite{tomar_efficient_2017} introduced a heap-based energy replenishment scheme, which efficiently prioritizes nodes based on their energy requirements. Smriti et al.~\cite{priyadarshani_efficient_2021} considered on-demand charging schemes that account for heterogeneous energy consumption and partial charging, optimizing the use of available energy resources.

These studies collectively aim to achieve efficient energy utilization and extend the lifespan of WRSNs by addressing various challenges associated with mobile charging, such as asymmetric routing, path optimization, energy constraints, and dynamic network conditions. We have summarized the research details of the aforementioned papers, as shown in Table~\ref{tab:summary of detail}.

\begin{table*}[!htpb]
\centering
\caption{Summary of Main Related Work}
\scriptsize
\label{tab:summary of detail}
\begin{tabular}{|>{\centering\arraybackslash}m{1.0cm}|>{\centering\arraybackslash}m{2.6cm}|>{\centering\arraybackslash}m{2.2cm}|>{\centering\arraybackslash}m{2.2cm}|>{\centering\arraybackslash}m{2.2cm}|>{\centering\arraybackslash}m{2.2cm}|>{\centering\arraybackslash}m{2.2cm}|}
\hline
\textbf{Work} & \textbf{Optimization objective} & \textbf{One-to-many charging} & \textbf{Charging position optimization} & \textbf{Support directional charging} & \textbf{Support mobile charging} & \textbf{Address asymmetric routing} \\ \hline
\cite{dai_optimizing_2017} & Charging utility & $\checkmark$ & $\checkmark$ & $\checkmark$ &  &  \\ \hline
\cite{wang_robust_2019} & Charging utility & $\checkmark$ & $\checkmark$ & $\checkmark$ & &    \\ \hline
\cite{jiang_efficient_2016} & Costs & $\checkmark$ & $\checkmark$ & $\checkmark$ &  &    \\ \hline
\cite{wu_charging_2019} & Charging utility & $\checkmark$ & $\checkmark$ &  & $\checkmark$ &   \\ \hline
\cite{liu_effective_2020} & Energy efficiency and Number of Dead node & $\checkmark$ & $\checkmark$ &  & $\checkmark$ &   \\ \hline
\cite{lin_maximizing_2020} & Charging utility & $\checkmark$ & $\checkmark$ &  & $\checkmark$ &   \\ \hline
\cite{dai_placing_2020} & Charging utility & $\checkmark$ & $\checkmark$ & $\checkmark$ & $\checkmark$ &    \\ \hline
\cite{tomar_efficient_2017} & Charging utility & $\checkmark$ & $\checkmark$ & $\checkmark$ & $\checkmark$ &   \\ \hline
\cite{priyadarshani_efficient_2021} & Charging Delay & $\checkmark$ & $\checkmark$ &  & $\checkmark$ &   \\ \hline
Our work & Energy loss and Time & $\checkmark$ & $\checkmark$ & $\checkmark$ & $\checkmark$ & $\checkmark$   \\ \hline
\end{tabular}
\end{table*}

\subsection{Mobile Charging in Asymmetric Environments}
The above work assumes routing-symmetric environments, but this assumption may not hold in many real-world scenarios, which hinders the applicability of the proposed algorithms. However, to date, very few studies have considered the routing-asymmetry feature in charging scheduling for WRSNs. For charging WRSNs in underwater environments, Lin et al.~\cite{Lin2018jss} considered the asymmetry of underwater movement, but the routing-asymmetry was only addressed within a simple greedy trajectory design framework, failing to be systematically considered and tackled. Furthermore, their work only deals with OMCs in the O2O charging mode, without exploring the more energy-efficient M2M mode, resulting in a less effective charging schedule. As far as we know, this is the only work that considers routing-asymmetry in addressing the charging scheduling problem in WRSNs.

In fact, routing-asymmetry has long been an important issue in various versions of trajectory design problems, such as the TSP problem. To facilitate solving ATSP using the powerful algorithms designed for TSP, a method was proposed in Ref.~\cite{Jonker1983} to transform an ATSP instance into a TSP instance. Better exploiting the state-of-the-art algorithms in the realms of TSP and ATSP can facilitate systematically addressing the charge scheduling problem in WRSNs with routing-asymmetry, making the solutions more applicable to real-world scenarios.

\section{Preliminary Models}
\label{sec:model}

In this section, we present the preliminary models, including the system model, asymmetric routing model, DMC's energy consumption model, WPT energy transfer model, and directional energy transfer coefficient model.

\subsection{System Model}

We consider an RA-WRSN consisting of a Base Station (BS) located at position $l_0$, a DMC, and $N$ rechargeable nodes in the set $\mathcal{U} = \{u_1,u_2,\ldots,u_N\}$. For simplicity, each node $u_i{\in}\mathcal{U}$ is also referred to by its position. The DMC is responsible for charging the nodes to ensure their normal operation. It typically recharges its battery fully at the BS and remains there until required to charge some nodes. Each node’s energy demand is referred to as a charging task. To address a set of charging tasks, the DMC departs from the BS, follows a tour to charge the nearby nodes, and then returns to the BS. The charging tour must be meticulously designed and scheduled to fulfill the charging tasks efficiently. An illustrative example of the RA-WRSN is shown in Fig.~\ref{fig2}.

\begin{figure}[htb]       
\centering
\includegraphics[width=0.49\textwidth]{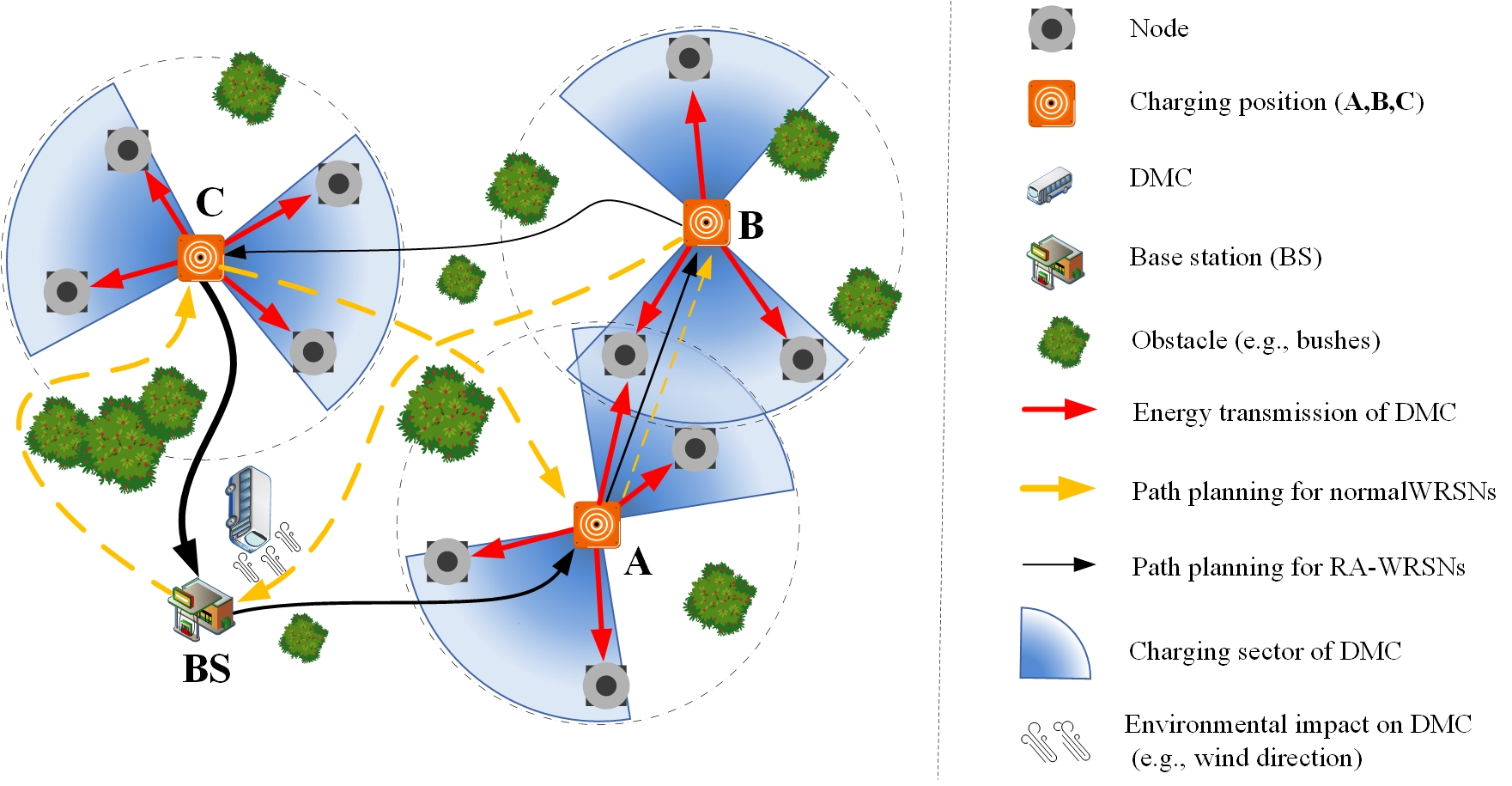}\\
\caption{An example RA-WRSN}
\label{fig2}
\end{figure}

Each node is equipped with a rechargeable battery and a wireless energy transceiver for transmitting and receiving energy. Node $u_i{\in}\mathcal{U}$ has some parameters as follows: initial energy $e_\text{B}^i$, energy demand $e_\text{D}^i$, and storage capacity $e_\text{C}^i$. Here, $e_\text{B}^i$ denotes the initial energy stored in the battery of $u_i$, and $e_\text{D}^i$ represents the amount of energy required by node $u_i$. To simplify notation, we compactly express the parameters of the nodes as column vectors. For example, we use $\textbf{e}_\text{B}$ to represent $\textbf{e}_\text{B}{:=}[e_\text{B}^1,e_\text{B}^2,\cdots,e_\text{B}^N]$. 

The DMC has the following parameters: energy transmission power $p_0$, initial energy $e_\text{B0}$, moving speed $\overline{v}$, and energy consumption base rate $w_0$ for moving one unit distance. Here we assume that the initial energy $e_\text{B0}$ is sufficient to meet the energy demands of all nodes. Otherwise, we can split the whole charging tour into several sub charging tours such that each sub-tour can be fully served by an DMC, as discussed in~\cite{Gao2023adhoc}.   

To simply the charge scheduling task, we additionally assume that the DMC does not transmit energy while moving. Therefore, a charging task set can be fulfilled in steps by firstly selecting charging positions in the network area, then determining a charging tour and charging schedule, to traverse these selected positions and staying their for a time to transmitting energy along some selected directions. 

\subsection{Asymmetric Routing Model}
The movement energy consumption incurred in moving the MC make up a main part of the energy consumption for fulfilling a charging schedule. Routing-asymmetry considerably affects the movement energy consumption of a charging tour. As introduced in Sec.~\ref{sec:intro}, routing-asymmetry results from various factors, such as terrain and topography restrictions, altitude difference, wind force if deployed in the air, and water flow if deployed on or in water. To abstract the heterogeneous underlying factors, we build an abstracting asymmetric routing model, which contains two parts: the asymmetric distance model and the asymmetric movement energy consumption rate model, where at the core are two factor coefficients for reflecting the effect of factors causing asymmetry. 

For the asymmetric distance model, we define routing-asymmetry (RA)  distance coefficient $k_\text{RA,Dis}{:=}k_\text{RA,Dis}(l_i,l_j|$ $ m_\text{RA},h_\text{RA},f_\text{RA})$. The RA distance coefficient mimics the asymmetric effect on the length of the path segment from $l_i$ to $l_j$, determined by the factors of $m_\text{RA}$, $h_\text{RA}$, and $f_\text{RA}$. Here $h_\text{RA}$ represents the height distance between $l_i$ and $l_j$, $f_\text{RA}$ represents the accumulated external force affecting the MC's movement from $l_i$ to $l_j$, and $m_\text{RA}$ represents the overall effect of other miscellaneous factors result from terrain and topography heterogeneity. Detailed formulation of $k_\text{RA,Dis}$ depends on the particular asymmetric environment, whereas the overall effect for all asymmetric environment finally lead to a real scale value. We abbreviate $k_\text{RA,Dis}(l_i,l_j|m_\text{RA},h_\text{RA},f_\text{RA})$ as $k_\text{RA,Dis}(l_i,l_j)$ for simplicity. Similarly, for the asymmetric movement energy consumption rate model, we define RA movement energy consumption rate coefficient $k_\text{RA,Egy}{:=}k_\text{RA,Egy}(l_i,l_j|m_\text{RA},h_\text{RA},f_\text{RA})$ to account for the asymmetry on actual movement energy consumption rate per unit distance from $l_i$ to $l_j$.

With the above defined coefficients, for any two positions $l_i$, $l_j$ in the area of the RA-WRSN, let $d_\text{ideal}(l_i,l_j)$ denote the ideal euclidean distance between them, and let $d(l_i,l_i)$ denote the routing-asymmetric real distance. To emphasize the context in routing-asymmetry, We call $d(l_i,l_i)$ as the \textbf{RA distance} of the path segment from $l_i$ to $l_j$. Similarly, we let $w(l_i,l_j)$ denote the routing-asymmetric real movement energy consumption rate per unit distance from $l_i$ to $l_j$, naming it as \textbf{RA energy consumption rate}. Then the RA distance and RA movement energy consumption of the path segment from $l_i$ to $l_j$ are modeled as 

\begin{align}
d(l_i,l_j)&=k_\text{RA,Dis}(l_i,l_j){\cdot}d_\text{Ideal}(l_i,l_j),\label{eq_ra_dis}\\
w(l_i,l_j)&=k_\text{RA,Egy}(l_i,l_j){\cdot}w_0.\label{eq_ra_egy}
\end{align}

The above two equations make up the asymmetric routing model. The routing-asymmetry are ultimately manifested as RA distances and RA energy consumption rates of directional path segments between charging position pairs. For notation simplicity, We collect the RA values into two matrices of $\textbf{D}{=}[d(l_i,l_j)]_{\{l_i,l_j{\in}\mathcal{L}\}}$ and $\textbf{W}{=}[w(l_i,l_j)]_{\{l_i,l_j{\in}\mathcal{L}\}}$. Here $\mathcal{L}$ is the set of charging positions.

\subsection{Energy Consumption Model of the DMC}
In fulfilling a charging schedule, the energy consumption of the DMC contains two parts: movement energy consumption for driving the DMC move, charging energy consumption resulting from charging the nodes. 

For a path segment from point $l_i$ to $l_j$ with RA distance $d(l_i,l_j)$ and movement energy consumption rate $w(l_i,l_j)$, the movement energy consumption $e^{\text{MC}}_\text{Move}(l_i,l_j)$ and the moving time $t^\text{MC}_\text{Move}(l_i,l_j)$ are respectively modeled as Eq.~\eqref{eq_dmconecost_egy} and Eq.~\eqref{eq_dmconecost_time}, where $\overline{v}$ is the constant moving speed of the DMC. 

\begin{align}
e^{\text{MC}}_\text{Move}(l_i,l_j)&=d(l_i,l_j)w(l_i,l_j),&l_i,l_j{\in}\mathcal{L}\label{eq_dmconecost_egy}\\
t^{\text{MC}}_\text{Move}(l_i,l_j)&=d(l_i,l_j)/\overline{v}, &l_i,l_j{\in}\mathcal{L}.\label{eq_dmconecost_time}
\end{align}

Let $\text{r}(\mathcal{L}){:=}[l_0,l_{\pi_1},l_{\pi_2},\ldots,l_{\pi_|\mathcal{L}|},l_{\pi_{(|\mathcal{L}|{+}1)}}{=}l_0]$ represent a charging tour contains $|\mathcal{L}|$ intermediate positions, then the movement energy consumption and moving time of the DMC following the tour can be obtained as

\begin{align}
e^{\text{MC}}_\text{Move}(\text{r}(\mathcal{L}))&=\sum_{i{=}0}^{|\mathcal{L}|}e^{\text{MC}}_\text{Move}(l_{\pi_i},l_{\pi_{(i{+}1)}}),\label{eq_dmctour_egy}\\
t^{\text{MC}}_\text{Move}(\text{r}(\mathcal{L}))&=\sum_{i{=}0}^{|\mathcal{L}|}t^{\text{MC}}_\text{Move}(l_{\pi_i},l_{\pi_{(i{+}1)}}).\label{eq_dmctour_time}
\end{align}

Let $\mathcal{S}_\text{Dir}(i){:=}\{\psi_1,\psi_2,\ldots,\psi_{k_i}\}$ denote the charging direction set at position $l_i$. We use a tuple $(l_i,\psi_j)$ to completely define a pair of a direction and its associated position. Without loss of generality, let assume that the DMC does not transmit energy at the BS, i.e., the number of charging directions at $l_0$ is $k_0{=}0$. Let $\mathcal{S}_\text{PosDir}(l_i){:=}\{(l_i,\psi_j)|i{\in}\{1,2,\ldots,k_i\}$, then $\mathcal{S}_\text{PosDir}{:=}$ ${\cup}_{i{=}1}^{|\mathcal{L}|}\mathcal{S}_\text{PosDir}(l_i)$ contains all position-direction (abbreviated as Pos-Dir) pairs. Let us sort the pairs in $\mathcal{S}_\text{PosDir}$ into a list firstly on position and then on direction value, and with a slight abuse of symbols, we re-use $\mathcal{S}_\text{PosDir}$ to denote the sorted list. 
We further assume that the charge time list (or column vector) corresponding to the Pos-Dir pairs in $\mathcal{S}_\text{PosDir}$ as $\textbf{t}^\text{MC}_\text{Tran}{=}[t^\text{Tran}_1,t^\text{Tran}_2,\ldots,t^\text{tran}_{K}]^\text{T}$ with $K{:=}\sum_{i{=}1}^{M}k_i$ denotes the length of list $\textbf{t}^\text{MC}_\text{Tran}$. With the assumption that the DMC always transmit energy with power $p_0$, the charging energy consumption of the DMC for transmitting energy can be obtained as

\begin{equation}
\label{eq_dmc_egy_loss_tran}
e^{\text{MC}}_\text{Tran}{:=}p_0{\cdot}{\sum}_{i{=}1}^{K}t^\text{Tran}_i
{=}p_0\mathds{1}^{1{\times}K}\textbf{t}^\text{MC}_\text{Tran}.
\end{equation}

The total energy consumption and the final energy of the DMC after fulfilling a charging schedule can be obtained as Eq.~\eqref{eq_dmc_egy_total} and Eq.~\eqref{eq_dmc_egy_final}, respectively.

\begin{equation}
\label{eq_dmc_egy_total}
e^{\text{MC}}_\text{Total}{=}e^{\text{MC}}_\text{Tran}{+}e^{\text{MC}}_\text{Move}(\textbf{r}(\mathcal{L})).
\end{equation}

\begin{equation}
\label{eq_dmc_egy_final}
e_\text{F0}{=}e_\text{B0}{-}e^{\text{MC}}_\text{Total}{=}e_\text{B0}{-}e^{\text{MC}}_\text{Tran}{-}e^{\text{MC}}_\text{Move}(\textbf{r}(\mathcal{L})).
\end{equation}

\subsection{WPT Energy Transfer Model}
\label{sec_wpt_egy_transfer_model}
Let $c(i,j)$ denote the energy transfer coefficient from Pos-Dir pair $i$ in $\mathcal{S}_\text{PosDir}$ to node $u_j$. If the DMC transmits energy at power $p_0$, the energy power received by node $u_j$ is expressed as $p(k){:=}c(i,j)p_0$. Let construct energy transfer coefficient matrix $\mathbf{C}{:=}[c(i,j)]_{\{i{\in}\{1,2,\cdots,K\},j{\in}\{1,2,\cdots,|\mathcal{L}|\}\}}$, then given the energy transmission time duration list $\textbf{t}^\text{MC}_\text{Tran}{=}[t_1,t_2,\ldots,t_{K}]^\text{T}$ corresponding to the Pos-Dir pairs in $\mathcal{S}_\text{PosDir}$, the list of energy received by all nodes $\textbf{e}_\text{R}{:=}[e^\text{R}_1,e^\text{R}_2,\cdots,e^\text{R}_M]^\text{T}$ can be compactly expressed as 

\begin{equation}
\label{eq_egy_received}
\textbf{e}_\text{R}{=}p_0\mathbf{C}\textbf{t}^\text{MC}_\text{Tran}.
\end{equation}

Considering the capacity limits of the nodes, final energy stored in the batteries of the nodes, denoted as $\textbf{e}_\text{F}$, can be expressed as Eq.~\eqref{eq_node_egy_final}, where $\min$ means element-wise minimization.

\begin{equation}
\label{eq_node_egy_final}
\textbf{e}_\text{F}{=}\min\{\textbf{e}_\text{B}{+}\textbf{e}_\text{R},\textbf{e}_\text{C}\}.
\end{equation}

\subsection{Directional Energy Transfer Coefficient Model}

Energy transfer coefficient model is used to determine the energy transfer coefficient used in Sec.~\ref{sec_wpt_egy_transfer_model}. Here we adopt a widely used directional energy transfer coefficient model as expressed in Eq.~\eqref{eq_dmc_egy_tranfer_coeff_model}, as outlined in~\cite{dai_radiation_2017}. Here the region covered by the energy transmission signal of a DMC is a sector rooted at the DMC with radius equals charge distance $D$ and sector angle $\varphi$, as depicted in Fig.~\ref{fig3}. The center-line of the coverage sector is referred to as the corresponding charging direction, denoted as $\psi$. With the additional parameter including charging direction angle $\psi$, node direction angle $\theta$, and the distance $d$ from the DMC to the node concerned, Eq.~\eqref{eq_dmc_egy_tranfer_coeff_model} determines the energy transfer coefficient from the DMC to the node. All direction angles in the model are measured relative to a certain reference direction with positive in counterclockwise. We assume that the DMC’s sector angle $\varphi$ remains constant, whereas $\psi$ is freely adjustable.

\begin{equation}
\label{eq_dmc_egy_tranfer_coeff_model}
c(\psi,\varphi,\theta,d ){=} 
\left\{
\begin{array}{ll}
\frac{\delta}{(\alpha{+}d)^{\beta}},&d{\le}D, \theta{\in}[\psi{-}\varphi/2,\psi{+}\varphi/2],\\
0,&\text{otherwise}. \\
\end{array}
\right.
\end{equation}

\begin{figure}[htb]       
\centering
\includegraphics[width=0.32\textwidth]{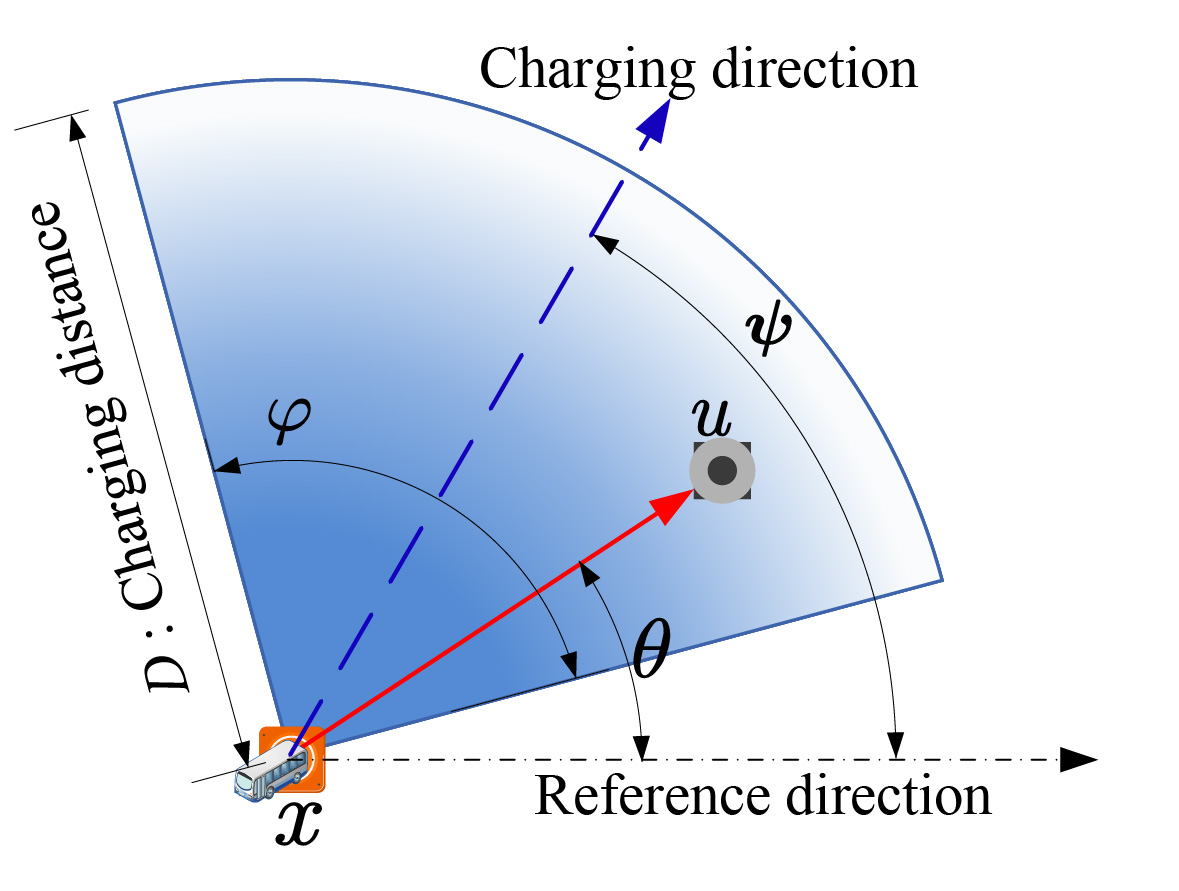}\\
\caption{Coverage sector of a DMC}
\label{fig3}
\end{figure}

The energy radiation leads to a charge sector is called a charging beam. As described in ~\cite{lee_energy-efficient_2021}, DMCs may be capable of generating multiple charging beams simultaneously. So they can correspondingly be classified as single-beam DMC and multi-beam DMC. We focus on single-beam DMC in this paper, whereas the multi-beam DMC case can be addressed following the idea in~\cite{gao_scheduling_2023}, where normal symmetrical WRSNs are assumed.

For ease of reference, we compiled the main symbols used in this paper in Table~\ref{Symbol Definitions}.

\begin{table}[!htpb]
\centering
\captionsetup{justification=centering}
\caption{\textsc{Symbol Definitions}}
\label{Symbol Definitions}
\begin{tabular}{|p{0.07\textwidth}|p{0.37\textwidth}|}
\hline
\textbf{Symbols}  & \textbf{Description}  \\ 
\hline
$\mathcal{U},u_i,N$   & $\mathcal{U}{:=}\{u_1,u_2,\ldots,u_N\}$ denotes the set of nodes with node number $M$, where $u_i$ is the $i$-th node.\\
\hline
$\mathcal{L},l_i,L$   & $\mathcal{L}{:=}\{l_1,l_2,\ldots,l_L\}$ denotes the set of charging positions with position number $L$, where $l_i$ is the $i$-th position $l_i$.\\
\hline
$\textbf{r}(\mathcal{L})$   & A charging tour on $\mathcal{L}$, i.e., a tour traverses all positions in $\mathcal{L}$.  \\ 
\hline
$e_\text{B},e_\text{D},e_\text{C}$   & $e_\text{B}{:=}\{e^\text{B}_1,\ldots,e^\text{B}_N\}$ is the list of initial energy of all nodes. $e_\text{D}$ and $e_\text{C}$ are the energy demanding list and storage capacity list.\\ 
\hline
$e_\text{F}, e_\text{R}$   &The lists of final energy and received energy of all nodes.\\ 
\hline
$e_\text{B0},e_\text{F0}$   &  DMC's initial energy and final energy.\\ 
\hline
$p_0,\overline{v},w$   & Energy transmission power, moving speed, energy consumption base rate of the DMC.  \\ 
\hline
$\psi,\varphi,D,\theta$   & Charging direction angle, charge sector angle, charge distance, and node direction angle.  \\ 
\hline
$e^\text{MC}_\text{Tran}(s)$, $e^\text{MC}_\text{Move}(s)$, $e^\text{MC}_\text{Loss}(s)$    & Charging energy consumption, movement energy consumption, and total energy loss of the DMC for fulfilling a charging schedule $s$. \\        
\hline
$e^\text{Total}_\text{Loss}(s)$, $e^\text{Nodes}_\text{Rcv}(s)$   & Total energy loss and total energy received by all nodes for fulfilling a charging schedule $s$. \\     
\hline
$e^\text{MC}_\text{Tran}(\textbf{t})$, $e^\text{Nodes}_\text{Rcv}(\textbf{t})$, $e^\text{WPT}_\text{Loss}(\textbf{t})$   & Total energy loss, total energy received by all nodes, and the charging energy loss for fulfilling a charging schedule with energy transmission time list $\textbf{t}$. \\     
\hline
$t^\text{MC}_\text{Tran}(s)$, $t^\text{MC}_\text{Move}(s)$  & Energy transmission time list and movement time list of the DMC in a charging schedule $s$. \\        
\hline
$\textbf{C},\textbf{D},\textbf{W}$   & Energy transfer coefficient matrix, RA distance matrix, RA movement energy consumption matrix.\\ 
\hline
$\textbf{s}_\text{MC}$   & A charging schedule consists of schedule items with structure $(state,l,\psi,t)$.    \\ 
\hline
$d(l_i,l_j)$, $w(l_i,l_j)$   & Routing-asymmetric distance and routing-asymmetric energy consumption of the path segment from position $l_i$ to position $l_j$. \\ 
\hline
$\mathcal{S}_\text{Dir}(l_i)$, $\mathcal{S}_\text{PosDir}$   & $\mathcal{S}_\text{Dir}(l_i)$ is the set of charging directions at position $l_i$, $\mathcal{S}_\text{PosDir}{:=}\cup_{l_i{\in}\mathcal{L}}\mathcal{S}_\text{Dir}(l_i)$.\\ 
\hline
\end{tabular}
\end{table}

\section{The ADMCCS Problem}
\label{sec:ADMCCS}

In this section, we first provide the structure of the charging schedule list, then define the ADMCCS problem, present its formulation, and prove that it is NP-hard.

\subsection{Structure of Charging Schedule List}
A charging schedule for fulfilling a charging task set involves the DMC's charging trajectory design, the charging tour time scheduling, and the determination of charging direction and charging time. We represent a complete charging schedule using an ordered list named the DMC's Operation Schedule (DOS). Here, we only consider schedules that, by default, satisfy all nodes' energy demands.

Let $\mathbf{s}_\text{MC}{=}[s^E_1,s^E_2,\ldots,s^E_{m_1}]$ represent a DOS list with $m_1$ items, where each item is a tuple $s^E_i{:=}(state,l,\psi,t)$. These items are classified into two categories: movement schedule items for arranging travel between charging positions, and energy transmission schedule items for arranging the energy transmission operation along directions at certain positions.
$state{\in}{0,1}$ is a binary variable indicating the category of the item. An item with $state{=}0$ means that the DMC should move toward position $l$ for a time duration $t$, and the field $\psi$ is not used. An item with $state{=}1$ means that the DMC should transmit energy at position $l$ along direction $\psi$ for a time duration $t$. For example, a schedule item $(1,3,45,5)$ means that the DMC, which is at position $l_3$, should transmit energy along the $45^{\circ}$ direction for 5 time units. A schedule item $(0,4,90,6)$ means that the DMC should move toward position $l_4$ for 6 time units.

Let $\mathcal{S}^\text{DOS}_\text{Move}$ and $\mathcal{S}^\text{DOS}_\text{Tran}$ denote the sets of movement schedule items and charging schedule items, respectively. Then, $\mathcal{S}^\text{DOS}_\text{Move}{=}\{s^E_i|s^E_i.state{=}0\}$ and $\mathcal{S}^\text{DOS}_\text{Tran}{=}\{s^E_i|s^E_i.state{=}1\}$.

For an item $s_i^E{\in}\mathbf{s}_\text{MC}$, we use $s_i^E.state$ to access the state value of the item. This rule applies to other fields of the schedule items. As a list, the items in $\mathbf{s}_\text{MC}$ are executed sequentially to fulfill the schedule. Let $\tau_{i}^\text{E}{:=}\sum_{j{=}1}^{i}s_j^\text{E}.t$, then item $s_i^E$ starts at time $\tau_{i{-}1}^\text{E}$ and ends at time $\tau_{i}^\text{E}$.
  
\subsection{The ADMCCS Problem}
For an instance of the ADMCCS problem, the total initial energy $e_\text{TB}{:=}e_\text{B0}{+}\mathds{1}^{1{\times}N}\textbf{e}_\text{B}$ is a constant. For a DOS charging schedule $\textbf{s}$, the total final energy of all nodes and the DMC is $e_\text{TF}(\textbf{s}){:=}e_\text{F0}{+}\mathds{1}^{1{\times}N}\textbf{e}_\text{F}$. Therefore, the total energy loss for fulfilling schedule $\textbf{s}$ is $e^\text{Total}_\text{Loss}(\textbf{s}){:=}e_\text{TB}{-}e_\text{TF}(\textbf{s})$, and the total energy consumption of the DMC after completing the schedule $\textbf{s}$ is $e^\text{MC}_\text{Loss}(\textbf{s}){:=}e_\text{B0}{-}e_\text{F0}$. We further define $e_\text{Rcv}^\text{Nodes}(\textbf{s}){:=}\mathds{1}^{1{\times}N}(\textbf{e}_\text{F}{-}\textbf{e}_\text{B})$. Thus, we have the total energy loss $e^\text{Total}_\text{Loss}(\textbf{s}){=}e^\text{MC}_\text{Loss}(\textbf{s}){-}e_\text{Rcv}^\text{Nodes}(\textbf{s})$. As $e^\text{MC}_\text{Loss}(\textbf{s})$ and $e^\text{Nodes}_\text{Rcv}(\textbf{s})$ are both completely determined by the energy transmission time list $\textbf{t}$ embedded in schedule $\textbf{s}$, they can also be expressed as $e^\text{MC}_\text{Loss}(\textbf{t})$ and $e^\text{Nodes}_\text{Rcv}(\textbf{t})$. Since $e^\text{MC}_\text{Loss}(\textbf{t}){-}e^\text{Nodes}_\text{Rcv}(\textbf{t})$ represents the energy loss resulting from energy transmission, we call it charging energy loss and define it as in Eq.~\eqref{eq_egy_loss_wpt}.

\begin{equation}
\label{eq_egy_loss_wpt}
e^\text{WPT}_\text{Loss}(\textbf{t}){:=}e^\text{MC}_\text{Loss}(\textbf{t}){-}e^\text{Nodes}_\text{Rcv}(\textbf{t}).
\end{equation}

The goal here is to minimize the total energy loss while satisfying the energy demands of all nodes. With this objective, our targeted problem in this paper, named the Asymmetric DMC Charging Schedule (ADMCCS) problem, can be formally stated as follows:

\textbf{ADMCCS Problem: Given an RA-WRSN consisting of a DMC, $\mathbf{N}$ nodes, and a BS at position $l_0$, where the nodes have parameters including initial energy list $\mathbf{e}_\text{B}$, energy demand list $\mathbf{e}_\text{D}$, and battery capacity list $\mathbf{e}_\text{C}$, and the DMC has parameters including energy transmission sector angle $\mathbf{\varphi}$, transmission power $\mathbf{p_0}$ (corresponding to charge distance $\mathbf{D}$), and move speed $\mathbf{\overline{v}}$ (m/s), the task is to find a charging schedule $\mathbf{s}_\text{CS}$ with minimum energy loss $\mathbf{e}^\text{Total}_\text{Loss}$ while ensuring that all charging demands are satisfied.}

The ADMCCS problem can be formulated as in Eq.~\eqref{eq_ADMCCS_problem_math}.

\begin{equation}
\label{eq_ADMCCS_problem_math}
\begin{array}{lll}
\textbf{(P1)}&\mathop{\min}\limits_{
\tiny
\makecell
{\mathcal{L},
\textbf{r}(\mathcal{L}),\mathcal{S}_\varphi(\mathcal{L}),
\textbf{t}^\text{MC}_\text{Tran}
}
}
&  e_\text{Loss}^{\text{Total}}(
\textbf{r}(\mathcal{L}),
\mathcal{S}_{\varphi}(\mathcal{L}),
\textbf{t}^{\text{MC}}_{\text{Tran}}
), \\
&\text{s.t}.&C1:\textbf{e}_\textbf{F}{\geq}\textbf{e}_\textbf{B}{+}\textbf{e}_\textbf{D};\\
&&C2:\textbf{t}^\text{MC}_\text{Tran}{\geq}\textbf{0};\\
&&C3: \mathcal{L}{\in}\Omega, \textbf{r}(\mathcal{L}){\in}\Omega_\mathcal{L},\mathcal{S}_\varphi{\in}[0,2\pi)
;
\end{array}
\end{equation}

In Eq.~\eqref{eq_ADMCCS_problem_math}, $\Omega$ denotes the solution space of charging positions, i.e., the region accessible for charging, and $\Omega_\text{r}(\mathcal{L})$ denotes the solution space of valid charging tours traversing all positions in $\mathcal{L}$. The interval $[0,2\pi)$ represents the space of charging directions. $C1$ ensures that all nodes' energy demands are satisfied. Here, for two vectors ${\mathbf{a}}{:=}[a_1,a_2,\ldots,a_n]$ and ${\mathbf{b}}{:=}[b_1,b_2,\ldots,b_n]$, ${\mathbf{a}}{\ge}{\mathbf{b}}$ means $a_i{\ge}b_i$ for all $i{\in}\{1,2,\ldots,n\}$.

In the objective function in Eq.~\eqref{eq_ADMCCS_problem_math}, the dependencies of $e_\text{Loss}^{\text{Total}}$ on $\textbf{r}(\mathcal{L})$, $\mathcal{S}_{\varphi}(\mathcal{L})$, and $\textbf{t}^{\text{MC}}_{\text{Tran}}$ are explicitly emphasized. 

\begin{theorem}
The ADMCCS problem is NP-hard.
\end{theorem}

\begin{proof}
We prove this by showing that a restricted version of ADMCCS is indeed the Asymmetric Traveling Salesman Problem (ATSP), which is known to be NP-hard \cite{Papadimitriou2006}.

To this end, consider the ADMCCS problem with the following additional constraints: The charging distance $D$ of the DMC is very small, i.e., $D{=}0$, such that the DMC can only charge a node exactly at the node's location. This constraint prohibits the M2M charging mode and allows only the O2O charging mode. Thus, the DMC has to charge the nodes one by one, providing exactly the amount of energy that satisfies their energy demands—no more, no less. In this situation, the charging order makes no difference in terms of charging energy consumption and the total energy received by the nodes. As a consequence, the quality of a charging schedule is solely determined by its movement energy consumption, which is completely determined by the charging trajectory of the DMC. 

Therefore, finding an optimal charging schedule with minimum energy loss for the restricted version of the ADMCCS problem corresponds to the task of finding an optimal charging trajectory of the DMC with minimum movement energy consumption, which is indeed the ATSP problem.
\end{proof}

\section{Solve the ADMCCS Problem}
\label{sec:Solving}

As the ADMCSS problem is NP-hard, we focus on designing efficient approximation algorithms. To achieve this, and inspired by \cite{gao_energy_2019}, we address the problem in four steps, as outlined in Fig.~\ref{fig1} in Sec.~\ref{sec:intro}.

\begin{enumerate}[\textbf{Step} 1:]
\item Select a minimum-size set of charging positions that ensures all nodes can be charged by the DMC from some of these positions.
\item Select a minimum-size set of functional-equivalent charging directions at the charging positions.
\item Determine the energy transmission time duration along the directions at all the charging positions.
\item Find an energy-efficient loop tour to traverse the charging positions in RA-WRSNs.
\end{enumerate}

\subsection{Select Minimum-Size Charging Position Set}

Considering the charge distance $D$, we define a node as being covered by a charging position if the distance between the node and the charging position is not greater than $D$. A charging position set that ensures all nodes are covered by at least one position in the set is called a valid charging position set.

In this step, our goal is to determine a valid charging position set with the minimum size and the smallest total sum of a distance metric. Let $\mathcal{L}$ denote a charging position set, and let $\Omega$ denote the solution space of charging position sets. The task in this step can be formulated as in Eq.~\eqref{eq_ADMCCS_step1_chrgpos}.

\begin{equation}
\label{eq_ADMCCS_step1_chrgpos}
\begin{array}{lll}
\textbf{(P2)}&\mathcal{L}^*{=}&\arg\mathop{\min}\limits_{\mathcal{L}} |\mathcal{L}|{+}\lambda\sum_{l_i{\in}\mathcal{L}}d^{\max}_{i}, \\
&\text{s.t}.&C4:d^{\max}_{i}{=}\mathop{\min}\limits_{l_j{\in}\mathcal{L}}d(u_i,l_j), u_i{\in}\mathcal{U}.\\
&&C5:d^{\min}_{i}{\leq}D, u_i{\in}\mathcal{U}.\\
&&C6:d^{\max}_{j}{=}\mathop{\max}\limits_{
\scriptsize 
\makecell{
u_i{\in}\{u_i|u_i{\in}\mathcal{U},\\d(u_i,l_j){\leq}D\}
}
}d(u_i,l_j), l_j{\in}\mathcal{L};\\
&&C7:\mathcal{L}{\in}\Omega_L.
\end{array}
\end{equation}

In Eq.~\eqref{eq_ADMCCS_step1_chrgpos}, the objective is a combination of the size of the charging position set and the sum of the distances between a charging position and the farthest node covered by it. The coefficient $\lambda$ adjusts the relative importance of the distance metric and the number of charging positions $|\mathcal{L}|$. Constraint $C4$ defines $d^{\max}_{i}$ as the smallest distance between node $u_i$ and a charging position in $\mathcal{L}$, and constraint $C5$ requires that each node must be covered by at least one charging position in $\mathcal{L}$. Constraint $C6$ defines $d^{\max}_{j}$ as the maximum distance from a node covered by a charging position $l_j{\in}\mathcal{L}$ to the position itself, and constraint $C7$ restricts the solution space.

\begin{theorem}
The problem \textbf{P2} is NP-complete.
\end{theorem}

\begin{proof}
By removing the sector term in the objective function of P1, problem P1 is reduced to a new problem, which is essentially the Unit Disk Cover (UDC) problem \cite{Fowler1981}. The UDC problem is stated as follows: In the Euclidean plane, given a set $P{=}\{p_1,p_2,\cdots,p_n\}$ of $n$ points and a set of $m$ unit disks with centers in a set ${L}{=}\{l_1,l_2,\cdots,l_m\}$, the task is to determine a minimum-size set $L^*{\subseteq}L$ such that all points in $P$ are covered by the union of the unit disks with centers in $L^*$. The UDC problem is NP-complete \cite{Fowler1981}; therefore, problem P2 is NP-complete.
\end{proof}

Since P2 is NP-hard, we solve it in two steps guided by its objective function: (1) find a valid charging position set with minimum size; (2) refine the charging positions in the set to minimize $\sum_{l_i{\in}\mathcal{L}^*}d^{\max}_{i}$.

The node set covered by a charging position can be regarded as a cluster, with the charging position serving as the cluster center. Thus, clustering methods can be employed to address the first step. We propose a K-means Charging Position Generation (KCPG) algorithm to return a valid charging position set in polynomial time. This algorithm automatically adjusts the position set size and ultimately obtains the clusters and their center points. The specific code is provided later in Section~\ref{sec:RA-DMCS}.

The second step involves determining the smallest enclosing circle that covers each node cluster, returning the center point and the diameter of the smallest circle. This can be solved using the well-known Welzl algorithm \cite{Welzl1991lncs} or by using mature software packages such as MATLAB or \texttt{SciPy}.

\subsection{Determine Minimum-Size Functional-Equivalent charging direction Set}

At each charging position $l_i{\in}\mathcal{L}^*$, the DMC should transmit energy along selected directions. The DMC can take any direction within the infinite continuous range $[0,2\pi)$. Given that the charge distance $D$ and the charge sector angle $\varphi$ of the DMC are constants, we can use the Create Minimum Functional Representative Direction Set (cMFRAS) algorithm from \cite{gao_scheduling_2023} to obtain a minimum-size discrete direction set that is functionally equivalent to the set $[0,2\pi)$. These directions are then used as the charging directions. The optimality of cMFRAS was established in \cite{gao_scheduling_2023}.

In this paper, following the process outlined in \cite{gao_scheduling_2023}, we determine the functionally equivalent direction set for each position and collect them into the set of Pos-Dir pairs $\mathcal{S}_\text{PosDir}$. Let us assume $|\mathcal{S}_\text{PosDir}|{=}K$. Each Pos-Dir pair is then regarded as a Virtual MC (VMC). As described in Sec.~\ref{sec_wpt_egy_transfer_model}, the energy transfer coefficients from the VMCs to the nodes are collected into an energy transfer coefficient matrix $C$, where each row corresponds to a VMC, and each column corresponds to a node.

\subsection{Determine Time Duration for all charging directions}

Once the set $\mathcal{S}_\text{PosDir}$ of Pos-Dir pairs has been determined, we need to determine the energy transmission times of the DMC along the specified charging directions at the corresponding charging positions in each Pos-Dir pair, or, for short, \textbf{along the Pos-Dir pairs}. The complete set of energy transmission times determines the amount of energy that the nodes in the RA-WRSN can receive.

Recall that $\textbf{t}{:=}[t_1,t_2,\ldots,t_{K}]^{\text{T}}$ denotes the column vector of energy transmission times along the Pos-Dir pairs in a charging schedule $\textbf{s}$. Since the movement energy consumption of the DMC is not affected by $\textbf{t}$, minimizing the total energy loss $e^\text{Total}_\text{Loss}(\textbf{s})$ is equivalent to minimizing $e^\text{MC}_\text{Tran}(\textbf{t}){-}e^\text{Nodes}_\text{Rcv}(\textbf{t})$. As the DMC transmits energy with a constant power $p_0$, minimizing the charging energy consumption of the DMC is identical to minimizing the sum of the elements in $\textbf{t}$. Therefore, the transmission times can be determined by solving the problem presented in Eq.~\eqref{eq_energy_step2_posdirtime}.

\begin{equation}
\label{eq_energy_step2_posdirtime}
\begin{array}{lll}
\textbf{(P3)}&\textbf{t}^*{=}&\arg\min\limits_{\textbf{t}}\mathds{1}^{1{\times}K}\textbf{t}, \\
&s.t.& C1,C2; \\
&&C8:\textbf{t}{\geq}\textbf{0};\\
&&C9:Eq.~\eqref{eq_egy_loss_wpt}.\\
\end{array}
\end{equation}

We use \textbf{Cplex} to solve this problem, and denote the solution as $\textbf{t}^*$. By combining $\textbf{t}^*$ with the Pos-Dir pairs in the set $\mathcal{S}_\text{PosDir}$, we can construct the energy transmission schedule item set $\mathcal{S}^\text{DOS}_\text{Tran}$.

\subsection{Asymmetric Path Planning in RA-WRSN}
\label{sec_ADMCCS_step4_ATSP}

The objective of this step is to determine the minimum movement energy consumption loop tour $\textbf{r}$ that traverses all charging positions in $\mathcal{L}^*$, which were determined in the first step. In RA-WRSNs, the RA distances and energy consumption rates of the bidirectional path segments between points determine the travel time and energy consumption of the loop path. Routing asymmetry is the main challenge in this step. Recall that the effects of the routing asymmetry factors are ultimately manifested in the RA distance matrix $\textbf{D}$ and the RA movement energy consumption matrix $\textbf{W}$.

Let $\textbf{r}_\pi{:=}[l_{\pi_0}{=}l_0,l_{\pi_1},\ldots,l_{\pi_L},l_0]$ represent a charge tour consisting of $L{+}2$ points. Let $\{\textbf{r}_\pi\}$ denote the position set contained in $\textbf{r}_\pi$, and let $|\textbf{r}_\pi|$ denote the number of points in $\textbf{r}_\pi$, i.e., $|\textbf{r}_\pi|{=}L{+}2$. The movement energy consumption of the tour $\textbf{r}_\pi$ is $e^\text{MC}_\text{Move}(\textbf{r}_\pi)$. The task of asymmetric path planning for RA-WRSNs in this step can be formulated as follows:

\begin{equation}
\label{eq_ADMCCS_step4_ATSP}
\begin{array}{lll}
\textbf{(P4)}&\textbf{r}_\pi^*{=}&\arg\mathop{\min}\limits_{\textbf{r}_\pi}  e^\text{MC}_\text{Move}(\textbf{r}_\pi) \\
&s.t. & C10: \{\textbf{r}_\pi\}{=}\mathcal{L}^*{\cup}\{l_0\}, \\
&& C11: |\textbf{r}_\pi|{\geq}|\mathcal{L}^*|{+}1, \\
\end{array}
\end{equation}

In Eq.~\eqref{eq_ADMCCS_step4_ATSP}, constraint $C10$ ensures that all positions in $\mathcal{L}^*{\cup}\{l_0\}$ are traversed by the charge tour $\textbf{r}_\pi$. Constraint $C11$ ensures that the tour $\textbf{r}_\pi$ visits each point in $\mathcal{L}^*$ at least once.

In a symmetric Euclidean routing environment, due to the triangle inequality, a tour usually passes through each node exactly once. However, in an RA-WRSN, the triangle inequality does not necessarily hold, so sometimes a tour with sub-loops may be more efficient than a non-sub-loop tour. Therefore, $C10$ allows tours with sub-loops, which makes problem P4 different from the ATSP.

\begin{theorem}
Problem \textbf{P4} is NP-hard.
\end{theorem}

\begin{proof}
By applying an additional restriction that does not allow sub-tours, we obtain a restricted version of Problem P4. This restricted version is simply the Asymmetric Traveling Salesman Problem (ATSP). Since the ATSP is NP-hard \cite{Papadimitriou2006}, Problem P4 is also NP-hard.
\end{proof}

To solve P4, we use the most efficient and powerful heuristic algorithm, known as LKH, which achieves state-of-the-art performance in both solution quality and running speed. LKH performs exceptionally well in solving the ATSP.

Instead of solving the ATSP directly, another typical approach is to first transform the ATSP into a TSP, which can be done using the method proposed in ~\cite{Jonker1983}, and then solve it using traditional TSP algorithms. However, this transformation doubles the number of points to be visited, which significantly challenges TSP algorithms and greatly hinders their performance. This drawback is validated by the experiments in Sec.~\ref{sec_atsp_vs_transorm_tsp}.

Let $\textbf{r}_\pi^*$ denote the energy-minimizing path obtained using LKH. Based on $\textbf{r}_\pi^*$, we can derive the movement time duration list $\textbf{t}^{MC}_\text{Move}{=}[d(l_{\pi_0},l_{\pi_1})/\overline{v},d(l_{\pi_1},l_{\pi_2})/\overline{v},\ldots,d(l_{\pi_L},l_0)/\overline{v}]$, thereby determining all movement schedule items $\mathcal{S}^\text{DOS}_\text{Move}$.

Additionally, based on $\textbf{r}_\pi^*$, we arrange the energy transmission schedule items $\mathcal{S}^\text{DOS}_\text{Tran}$. For the schedule items corresponding to the directions of a single charging position, the items are sorted in ascending order of direction angles.

By combining the items in $\mathcal{S}^\text{DOS}_\text{Move}$ and $\mathcal{S}^\text{DOS}_\text{Tran}$, we can easily construct a charging schedule $\textbf{s}_{\text{MC}}$ as the final solution to the ADMCCS problem.

\section{The RA-DMCS Algorithm}
\label{sec:RA-DMCS}

Based on the analyses in Sec.~\ref{sec:ADMCCS}, we propose a heuristic algorithm, named RA-DMCS, to approximately solve the ADMCCS problem. As shown in the pseudocode in Alg.~\ref{alg_RA-DMCS}, RA-DMCS accepts inputs such as $e_\text{B}$, $e_\text{C}$, and generates a charging schedule $\textbf{s}_\text{MC}$. As outlined in Fig.~\ref{fig1}, RA-DMCS consists of four steps. It selects charging positions in step 1 (code line~\ref{alg_RA-DMCS_line_step1}), and determines charging directions in step 2 (code line~\ref{alg_RA-DMCS_line_step2}). In step 3, it formulates the problem following Eq.~\eqref{eq_energy_step2_posdirtime} (code line~\ref{alg_RA-DMCS_line_step3_1}), determines charging times for all Pos-Dir pairs (code line~\ref{alg_RA-DMCS_line_step3_2}), and constructs the schedule item set $\mathcal{S}^\text{MC}_{Tran}$ (code line~\ref{alg_RA-DMCS_line_step3_3}). In step 4, it generates a charging tour $\textbf{r}^*$ using the LKH algorithm (code line~\ref{alg_RA-DMCS_line_step4_lkh}) and constructs the schedule item set $\mathcal{S}^\text{MC}_{Move}$ (code line~\ref{alg_RA-DMCS_line_step4_scheduleset}). Finally, it constructs $\textbf{s}_\text{MC}$ and returns the result.

\begin{algorithm}
\caption{The RA-DMCS algorithm}
\label{alg_RA-DMCS}
\begin{algorithmic}[1]
    \item[\hspace{0.5em}\textbf{Input:}] $U$, $e_\text{B}$, $e_\text{D}$, $e_\text{C}$, $p_0$, $e_\text{B0}$, $D$, $\varphi$, $\overline{v}$; 
    \item[\hspace{0.5em}\textbf{Output:}] $s_{MC}$;
    \State Initial: $C$,\ $S_{dir}$;
    \State $\mathcal{L}^*{=}KCPG(U,D)$; \label{alg_RA-DMCS_line_step1}\Comment{Determine a charging position set, detailed in Alg.~\ref{alg_RA-DMCS}}
    \State Construct Pos-Dir pair set $\mathcal{S}_\text{PosDir}$ following the process in \cite{gao_scheduling_2023} by using the cMFRDS algorithm;\label{alg_RA-DMCS_line_step2}
    \State Construct the problem following Eq.~\eqref{eq_energy_step2_posdirtime};\label{alg_RA-DMCS_line_step3_1}
    \State Obtain $\textbf{t}^*$ using CPLEX;\label{alg_RA-DMCS_line_step3_2}
    \State Construct $\mathcal{S}^\textbf{MC}_\text{Tran}$ based on $\mathcal{S}_\text{PosDir}$ and $\textbf{t}^*$;\label{alg_RA-DMCS_line_step3_3}
   \State $\textbf{r}^*{=}LKH(\mathcal{L}^*,D,W)$;\label{alg_RA-DMCS_line_step4_lkh}
    \State Construct $\mathcal{S}^\text{MC}_\text{Move}$ based on $\textbf{r}^*$ as introduced in Sec.~\ref{sec:ADMCCS};\label{alg_RA-DMCS_line_step4_scheduleset}
    \State Construct $\textbf{s}_\text{MC}$ by combing $\mathcal{S}^\text{MC}_{Move}$ and $\mathcal{S}^\text{MC}_{Tran})$\label{alg_RA-DMCS_line_scheduleset};
    \State \textbf{return} $\textbf{s}_{MC}$;
\end{algorithmic}
\end{algorithm}

\subsection{Using LKH to Solve P2}

We adapted LKH for the ATSP problem to solve P4. LKH \cite{lkh_2019} was originally designed for solving standard TSP problems and was later adapted to the ATSP problem. According to \cite{lkh_web}, although LKH is heuristically approximate, computational experiments have demonstrated its high efficiency. As stated in \cite{lkh_web}, optimal solutions are produced with an impressively high frequency.

LKH typically starts with a heuristic initial solution, then tries to converge quickly to a high-quality solution by repeatedly applying k-opt sub-tour swaps. A k-opt sub-tour swap in LKH adjusts the tour by replacing $k$ path segments in the tour with $k$ new ones. Each iteration involves searching for the best k-opt sub-tour swap that most effectively reduces the total tour length. A crucial component of LKH is its effective swap quality evaluation method, which efficiently assesses the potential improvement of a k-opt sub-tour swap before it is applied, prioritizing the swaps with higher values to ensure significant reductions in tour length with each swap. LKH dynamically adjusts the value of $k$, allowing it to explore a larger solution space and escape local optima. It also dynamically adjusts the type of k-opt sub-tour swaps to perform, makes instance-specific adaptations, and thus effectively balances exploration and exploitation.

\subsection{KCPG Algorithm}

We propose a K-means Charging Position Generation (KCPG) algorithm to solve Eq.~\eqref{eq_ADMCCS_step1_chrgpos}. The pseudocode for KCPG is shown in Alg.~\ref{alg_KCPG}. The algorithm uses $j$ to denote the number of clusters, initializing it as 1. The K-means clustering algorithm is utilized to group the nodes into $j$ clusters (code line~\ref{alg_codeline_kmeans}). For each cluster, the center of the minimum circle covering the node set in the cluster is obtained (code line~\ref{alg_codeline_mincicle}). More clusters are constructed unless all minimum circles have radii smaller than the charge distance $D$. Finally, the centers of the clusters are returned as $\mathcal{L}$.

\begin{algorithm}
\caption{The KCPG algorithm}
\label{alg_KCPG}
\begin{algorithmic}[1]
    \item[\hspace{0.5em}\textbf{Input:}] Node set $U$, charging distance $D$;
    \item[\hspace{0.5em}\textbf{Output:}] The charging position set $\mathcal{L}$.
    \State Initialize $j{=}1$;\label{alg2_line_status}
    \While{true}
        \State $G{=}$ K\-means$(U,j)$;\label{alg_codeline_kmeans}
        \State $\{\mathcal{L}, R_M\}{=}$ Optimiz$(G)$;\label{alg_codeline_mincicle}
        \State \textbf{return} $\mathcal{L}$ if $r{\leq}D\hspace{2mm} {\forall}r{\in}R_M$ otherwise $j{=}j{+}1$;
    \EndWhile
\end{algorithmic}
\end{algorithm}

\begin{theorem}
The time complexity of the KCPG algorithm is $O(n^3)$.
\end{theorem}
\begin{proof}

KCPG iteratively calls the K-means clustering algorithm with an increasing number of clusters initialized as 1. In each iteration for a certain cluster number, K-means is used to cluster the nodes, and then the minimum enclosing circles for covering the nodes in each cluster are obtained using the Welzl algorithm~\cite{Welzl1991lncs}, and the results are checked to determine if the results are acceptable, otherwise a new iteration with the number of clusters increased by 1 is conducted.

According to ~\cite{macqueen_methods_1967}, K-means itself has a time complexity of $O(tkn)$, where $t$ is the number of iterations, $k$ is the number of clusters, $n$ is the number of data points, and $m{=}2$ is the dimension of the data points. The welzl algorithm~\cite{Welzl1991lncs} has time complexity $O(n/k)$, where $n/k$ is the size of a cluster. Running welzl for $k$ clusters requires time $O(k{\cdot}n/k)=O(n)$. Thus, an iteration with $k$ clusters has time complexity $O(ntk){+} O(n)$.

In the worst-case, cluster number will eventually increase to $n$. Given that $t$ and $t_m$ are constants, the time complexity of KCPG can be obtained as 

\begin{equation}
\label{KCPG_complexity_1}
\begin{array}{ll}
T&=O(n){\cdot}(O(ntk){+}O(n)){=}O(n^2tk{+}n^2)\\
&=O(n^2k{+}n^2) = O(n^3).
\end{array}
\end{equation}

\end{proof}

\section{Performance Evaluation}
\label{sec:simula}
In this section, we comparatively evaluate the performance of RA-DMCS against some representative algorithms through experiments. The experiments are coded with python and conducted on a computer with AMD R7-6800HS CPU, 16GB RAM, and Ubuntu 20.04.5 LTS OS.

\subsection{Algorithms for Comparison}
The algorithms in ~\cite{Lin2018jss} and in ~\cite{liang_grouping_2023} are selected for comparison. They are designated as One-to-One Greedy (O2OGre) and Greedy Grouping with Fixed-Direction Ant Colony (GFDA), respectively. 

O2OGre\cite{Lin2018jss} adopts a one-to-one charging approach, not exploiting the opportunity of multiple node simultaneous energy reception from the same energy signal. In O2OGre, the MC's charge tour decision problem is regarded as a symmetric TSP problem and solved using greedy algorithm.

GFDA~\cite{liang_grouping_2023} also adopts a grouped directional charging approach similar to us, but it is devoted to routing symmetric WRSNs. It begins by grouping the nodes using a greedy algorithm prioritized with energy demanding amount. Within each group, the center of the circle is selected as the charging position, and a certain number of charging directions are selected uniformly from the $[0,2\pi)$ range. The MC's charge tour decision problem is also regarded as a symmetric TSP problem, yet solved using an ant colony algorithm.

\subsection{Performance Metrics and Simulation Setup}
Four main performance metrics are employed: total energy loss, tour distance, and time span. The energy loss metric represents the total energy loss consisting of charging energy loss and movement energy consumption. The tour distance metric represents the total distance of the charging tour in a charging schedule. The time span metric of a charging schedule is the sum of the time field of all schedule items. 


Main simulation parameters and their default values are shown in Table~\ref{tab2}. A particular set of the values of the parameters is referred as a simulation setup. The impacts of the parameters on algorithms are investigated by experiments with various simulation setups, which differ only in the value of the parameter being inspected.

\begin{table}[htbp]
\begin{center}
\renewcommand{\arraystretch}{1.3} 
\caption{Simulation Parameters}
\label{tab2}
\centering
\begin{tabular}{>{\centering\arraybackslash}p{1.2cm}|>{\centering\arraybackslash}p{1.8cm}|p{1.2cm}|p{1.8cm}}
\toprule
\textbf{Symbols} & \textbf{Values} & \textbf{Symbols} & \textbf{Values} \\
\midrule
$N$ & 200 & $l_{BS}$ & (100, 100) \\
$\delta $ & 4000 & $\overline{v}$ & 1 m/s\\
$e_\text{C}$ & $60\sim 90$\ J & $k_\text{RA,Dis} $ & $0.5\sim 1.5$ \\
$e_\text{B}$ & $6\sim 36$\ J & $w$ & 4 J/m\\
$e_\text{D}$ & $18{\sim}75$\ J &  $D$ & 20 m\\
$p_0$ & 4 $W_0$ & $\varphi $ & $\pi$/4 \\
$\alpha $ & 100 & $L$ & 200 m  \\
$\beta $ & 2 & $k_\text{RA,Egy}$ & 1 \\
\bottomrule
\end{tabular}
\end{center}
\end{table}

To simulate routing-asymmetry, the RA distance coefficients and RA movement energy consumption rate coefficients defined in Sec.~\ref{sec_wpt_egy_transfer_model} are determined randomly in range [0.5,1.5]. 

To reduce randomness in results, simulations for each simulation setup are repeated for 200 problem instances. The nodes in these instances are randomly distributed within a 200m${\times}$200m area. The performance metrics are averaged over the 200 simulations to obtain the final results, and the 95\% confidence intervals are also calculated.

\subsection{Charge Tour Comparison for a Toy Network}

We visually illustrate the differences between the algorithms by applying them to a toy network with 10 nodes in a $100{\times}100$ area and a BS at position [50,50]. For space limitation, only the RA movement energy consumption rate coefficients for this network are shown in Table~\ref{tab:table3}.

\begin{table}[htbp]
\small
\caption{RA movement energy consumption rate coefficients for the toy network}
  \label{tab:table3}
  \centering
  \renewcommand{\arraystretch}{1.5} 
  \begin{tabular}{c|cccccc} 
    \toprule 
    & \multicolumn{1}{c|}{BS} & \multicolumn{1}{c|}{A} & \multicolumn{1}{c|}{B} & \multicolumn{1}{c|}{C} & \multicolumn{1}{c}{D}\\
    \midrule 
    BS & 0.00 & 1.11 & 1.10 & 1.14 & 1.35\\
    \hline
    A & 0.89 & 0.00 & 1.03 & 1.31 & 1.43\\
    \hline
    B & 0.90 & 0.97 & 0.00 & 1.33 & 1.15\\
    \hline
    C & 0.86 & 0.69 & 0.67 & 0.00 & 1.11\\
    \hline
    D & 0.65 & 0.57 & 0.85 & 0.89 & 0.00\\
    \bottomrule 
  \end{tabular}
\end{table}

The charge tours generated by the algorithms are shown in Fig.~\ref{fig5}, with sub-figures (a), (b), and (c) show the tours determined by RA-DMCS, GFDA, and O2OGre, respectively. Fig.~\ref{fig5}(b) provides the legends. In the sub-figures, the black dashed lines represent the DMC's trajectory, with the number on each path indicating the corresponding movement energy consumption. A green poly-lines indicates the direction of an energy transmission, with the attached number represents the corresponding energy received by the nodes in that direction. The four positions designated as A\~D are the selected charging positions. RA-DMCS exploits the routing-asymmetry property and generates a path outperforms the others. Comparatively, both O2OGre and GFDA could not deal with the routing-asymmetry characteristics well as designed for routing-symmetric environment. 

\begin{figure} 
  \centering
  \includegraphics[width=0.47\textwidth]{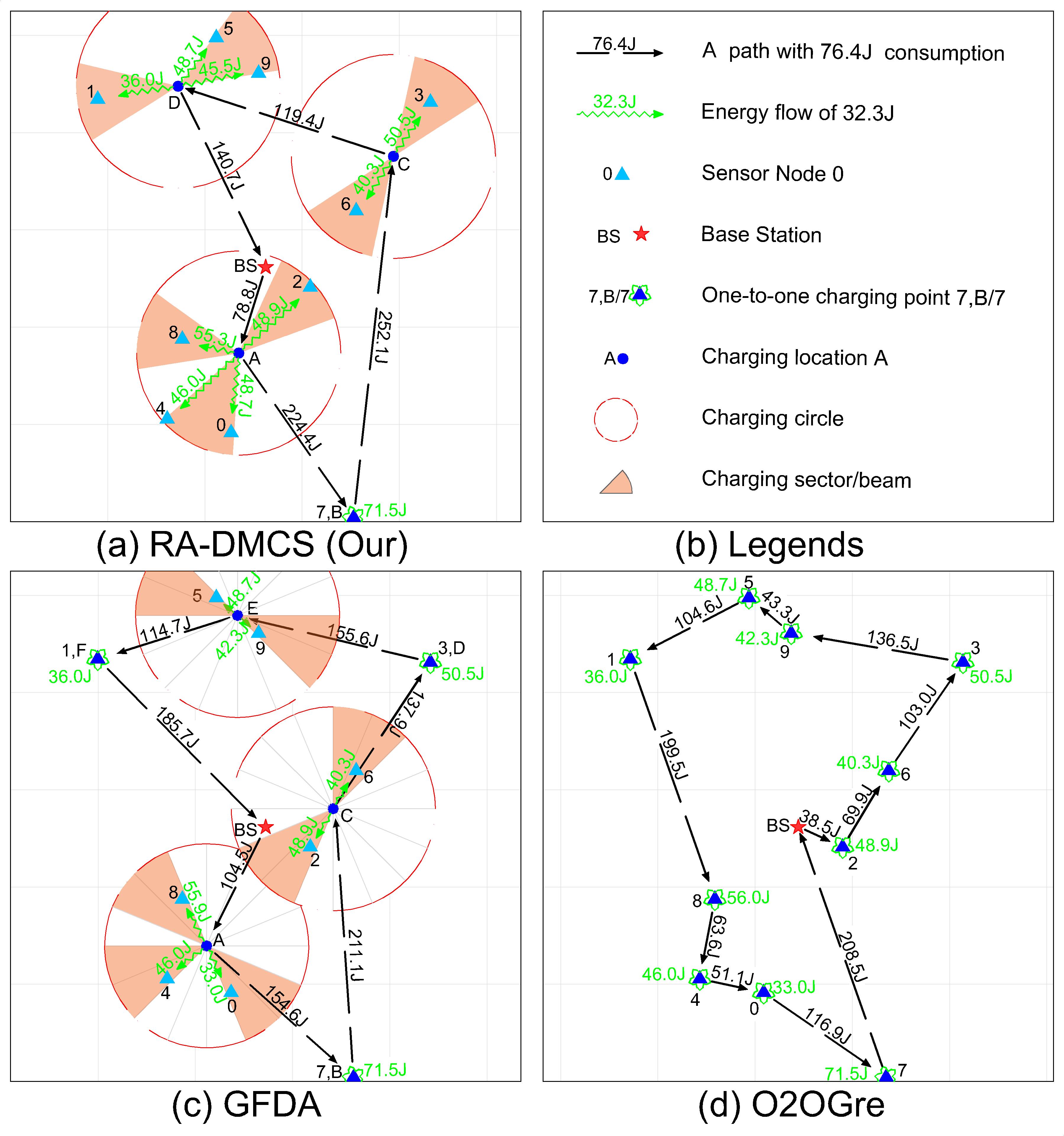}\\
  \caption{Charging tours determined by the algorithms}
  \label{fig5}
\end{figure}

Table~\ref{tab:table4} presents the detailed performance metrics of the algorithms for the network. Notably, RA-DMCS demonstrates lower energy consumption, higher charging efficiency, and shorter time span.

\begin{table}[htbp]
  \caption{Performance of the algorithms for the toy network}
  \label{tab:table4}
  \centering
\small
\renewcommand{\arraystretch}{1.5} 
  \begin{tabular}{>{\centering\arraybackslash}p{3.2cm}|>{\centering\arraybackslash}p{1.4cm}|>{\centering\arraybackslash}p{1.1cm}|>{\centering\arraybackslash}p{1cm}}
    \toprule 
    \textbf{Performance Metric} & \textbf{O2OGre}
    ~\cite{Lin2018jss} & \textbf{RA-DMCS}
    (Our) & \textbf{GFDA}
    ~\cite{liang_grouping_2023} \\
    \midrule 
    Energy loss (J) & 1844.82 & 1583.87 & 1912.02 \\
    \hline
    Time span (s) & 579.43 & 524.63 & 596.23 \\
    \hline
    Charging energy loss (J) & 709.38 & 768.51 & 847.93 \\
    \hline
    Movement energy consumption (J) & 1135.44 & 815.36 & 1064.09 \\
    \hline
    Charging time (s) & 295.57 & 312.78 & 330.21 \\
    \hline
    Moving time (s) & 283.86 & 211.85 & 266.02 \\
    \bottomrule 
  \end{tabular}
\end{table}

\subsection{Simulation Results and Analyses}

We conducted simulation experiments to examine the effects of key parameters on the performance of the algorithms. In these experiments, we make $N$ increase from 50 to 450 with an increment of 50, whereas the other parameters take their default values as in Table~\ref{tab2}. The results are shown in Fig.~\ref{fig6}. Only the results showing the effects of the number of nodes are provided here for space limitation.

\begin{figure*}[!htb]
\centering
    \subfigure[Energy loss]{\includegraphics[width=0.315\textwidth]{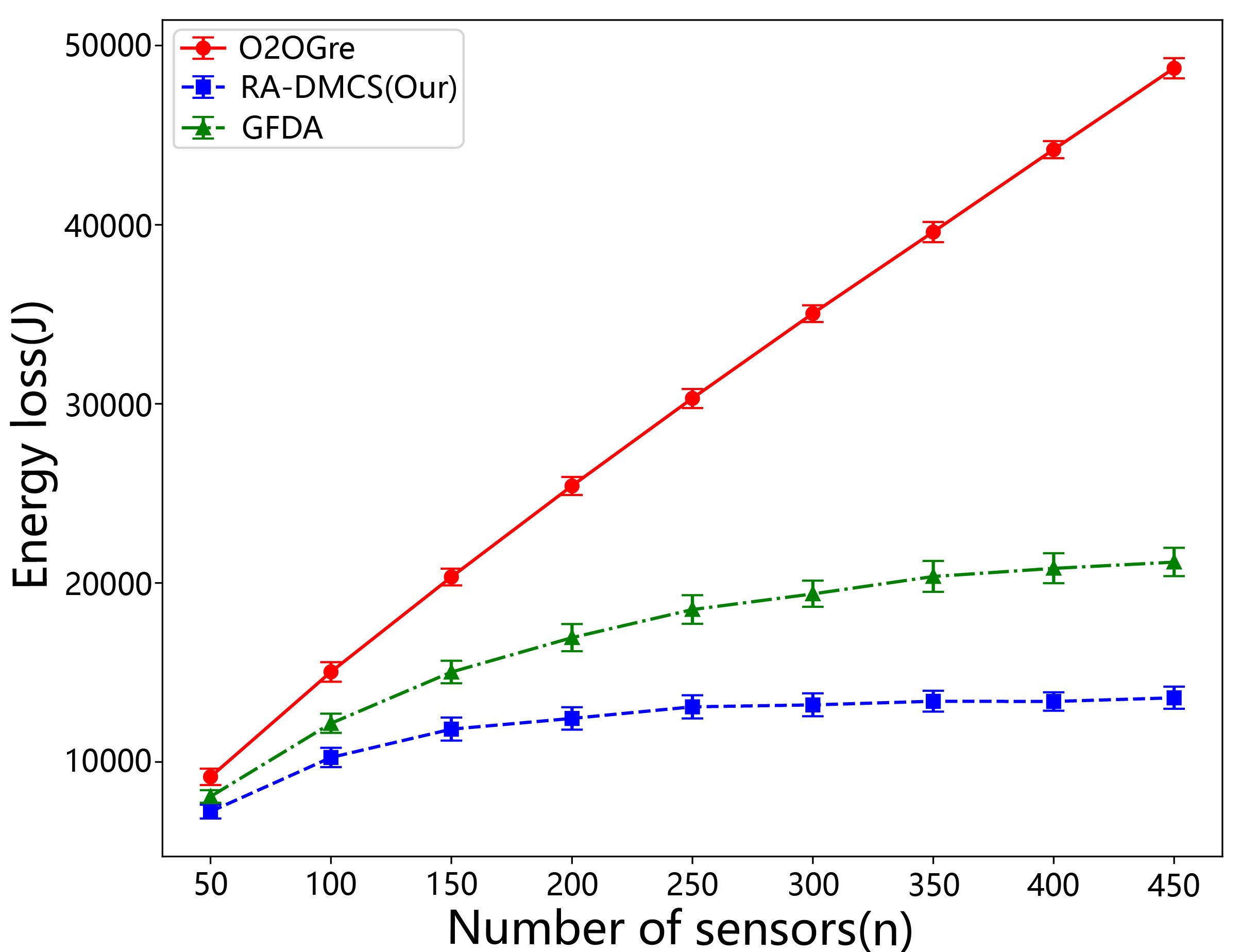}}
    \label{fig6(a)}
    \subfigure[Time span]{\includegraphics[width=0.32\textwidth]{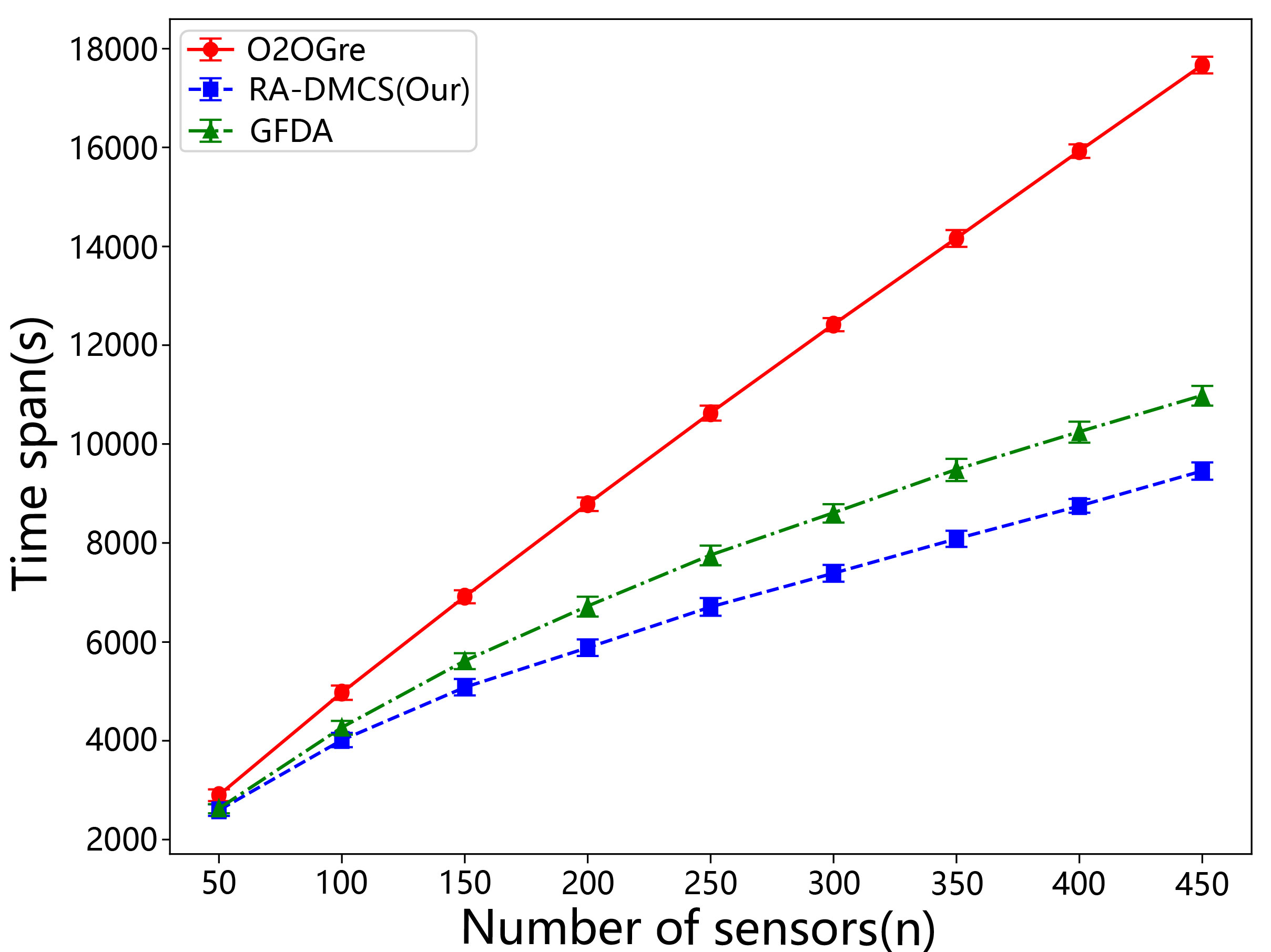}}
    \label{fig6(b)}
    \subfigure[Tour distance]{\includegraphics[width=0.32\textwidth]{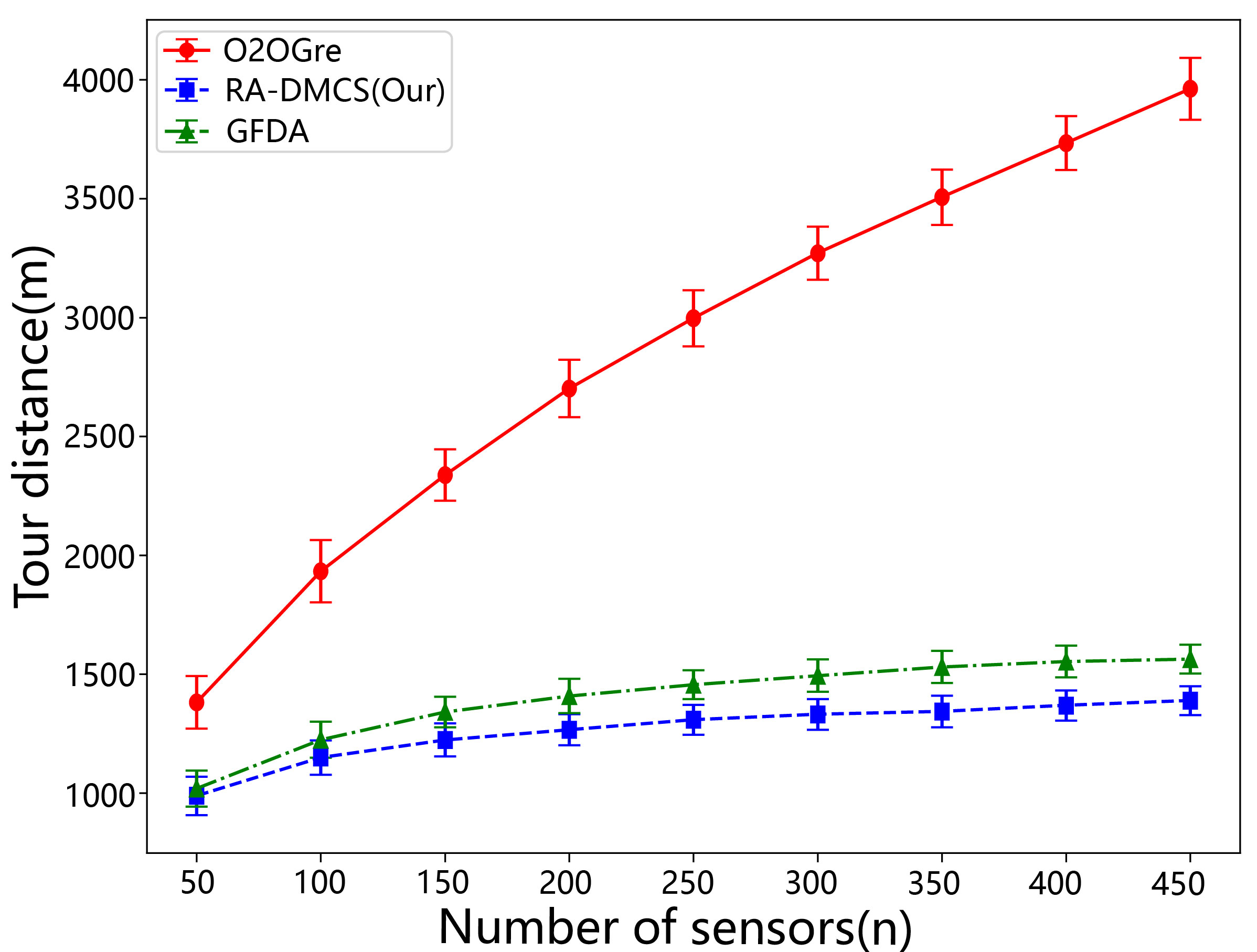}}
    \label{fig6(c)}
\caption{Effects of the number of nodes}
\label{fig6}
\end{figure*}

The results in Fig.~\ref{fig6}(a) shows that, as the number of nodes increases, the energy loss of the three algorithms also increases, yet RA-DMCS consistently maintains a lower energy consumption level than the other two algorithms. As the number of nodes increases, energy consumption of GFDA and RA-DMCS both tend to stabilize. This comes from two fold effects of increasing node number. On the one hand, compared to the large increase in node number, the number of clusters does not increase as much, which in turn makes the movement energy consumption does not increase much. On the other hand, as the number of nodes increases, each cluster will contain more nodes, and more energy will be harvested by them from the DMC's energy transmission signal, thus leading reduced energy loss. As an integrated effect, the energy consumption of the algorithms become more stable.

As shown in Fig.~\ref{fig6}(b), as the number of nodes increases, the time span of the solutions returned by the algorithms all increase. Compared with O2OGre, time spans of GFDA and RA-DMCS are much lower, as they both exploit M2M charging mode. Compared to O2OGre, Exploiting M2M not only reduces the number of charging positions, which leads to reduced move time and movement energy consumption, but also allows multiple nodes to be charged simultaneously, significantly shortening charging time. Compared with GFDA, our RA-DMCS employs an optimal direction selection algorithm and a better position selection method, making it outperforms GFDA. 

Fig.~\ref{fig6}(c) comparatively show the tour lengths of solutions obtained by the three algorithms. O2OGre utilizes the less efficient O2O charging mode, causing the DMC to visit all nodes to complete the entire charge task set, thus leading to a much longer tour length to complete the entire scheduling task. RA-DMCS also outperforms GFDA in term of this metric. Thus, RA-DMCS demonstrates the best performance among them in all the three performance metrics.

To better understand the performance of the algorithms, we further inspect the performance results in some detailed metrics, as shown in Fig.~\ref{fig7}. 

Fig.~\ref{fig7}(a) provides the results of energy consumption related metrics including total energy loss, charging energy loss, movement energy consumption, with an additional metric of tour distance for facilitating the inspection. As the travel distance metric scales differently from the other metrics, its y-axis is separately placed on the right side of the figure chart. The results show that RA-DMCS outperforms the others in the terms of these energy related metrics.

Fig.~\ref{fig7}(b) shows the results for time-related detail metrics including time span, charging time, moving time and algorithm runtime. RA-DMCS consumes the least time in terms of charging and moving times. RA-DMCS significantly outperforms GFDA and O2OGre. The results validate the superiority of RA-DMCS over others.

\begin{figure*}[!htb]
\centering
\subfigure[Energy loss]{\includegraphics[width=0.46\textwidth, height=0.3\textwidth]{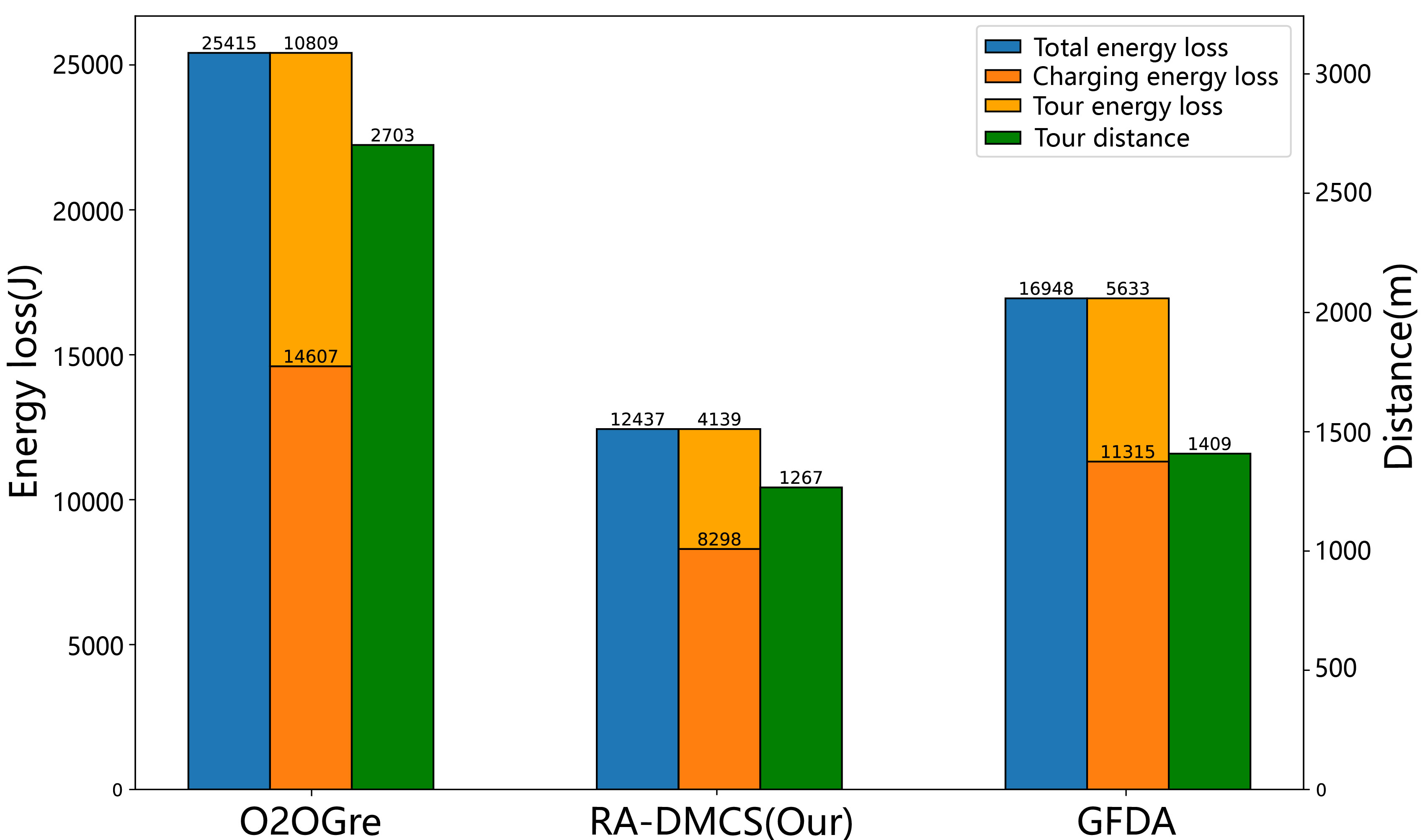}}
\subfigure[Time]{\includegraphics[width=0.46\textwidth,height=0.3\textwidth]{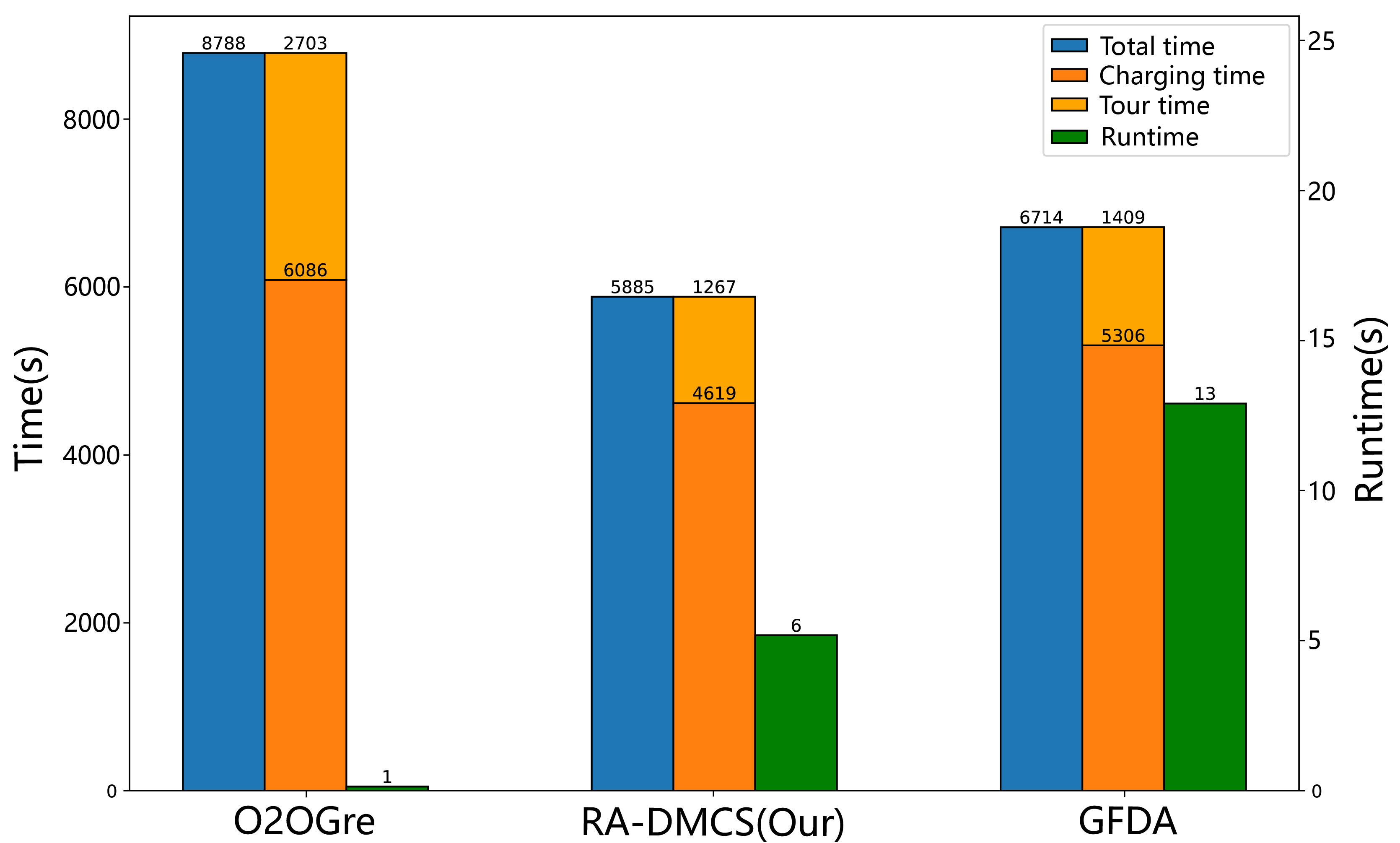}} 
\caption{Results of detail metrics of the algorithms}
\label{fig7}
\end{figure*}

\subsection{Comparison of Different Approaches in solving Charging tour decision problem P4}
\label{sec_atsp_vs_transorm_tsp}

Charging tour decision problem P4 in Eq.~\eqref{eq_ADMCCS_step4_ATSP} for routing-asymmetric WRSNs in the fourth step in Sec.~\ref{sec_ADMCCS_step4_ATSP} is a main challenge that distinguish our work from the literature. To concentrate on the performance of different approaches in solving the the charging tour decision problem, we select other five representative algorithms and test there performance in solving ATSP problem instances. 

For representing heuristic algorithms, we select an ant colony based algorithm and an greedy algorithm, which are designated as ATSP\_Ant and ATSP\_Greedy, respectively. For the approach of transforming ATSP to TSP and then solve it using TSP algorithms, here we reuse LKH, ant colony, and greedy algorithm to solve the transformed TSP problem. We denote the corresponding entire method as TSP\_LKH, TSP\_Ant, and TSP\_Greedy, respectively. For the original LKH algorithm that directly applies to ATSP, we denote it as ATSP\_LKH for discrimination. 

Simulation experiments similar to those in the previous section are conducted. We just provide the results for inspecting the effects of node number. The results of these six algorithms are shown in Fig.~\ref{fig8}.

\begin{figure*}[!htb]
\centering
\subfigure[Movement energy consumption]{\includegraphics[width=0.32\textwidth]{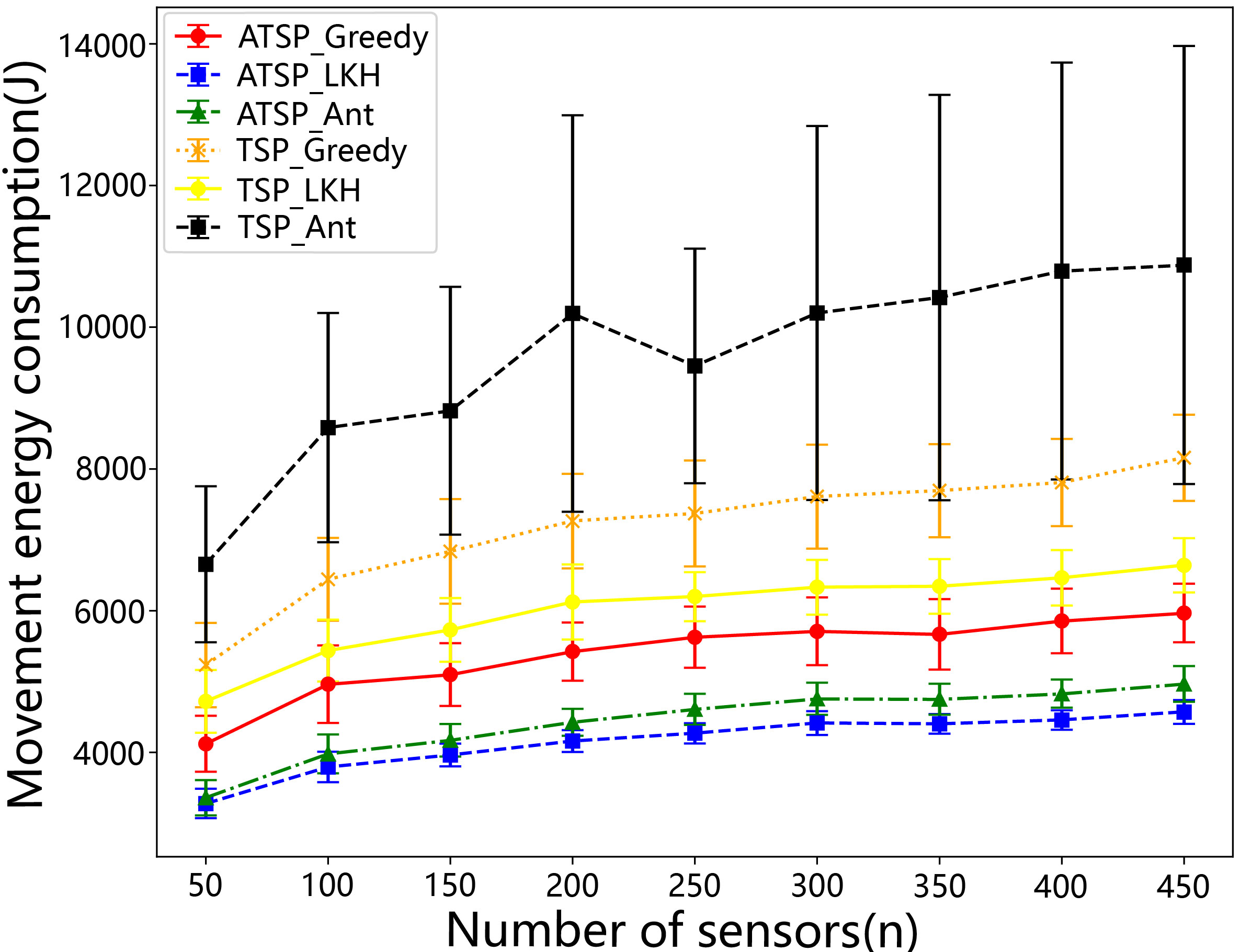}}
\label{fig8(a)}
\subfigure[Moving time]{\includegraphics[width=0.33\textwidth]{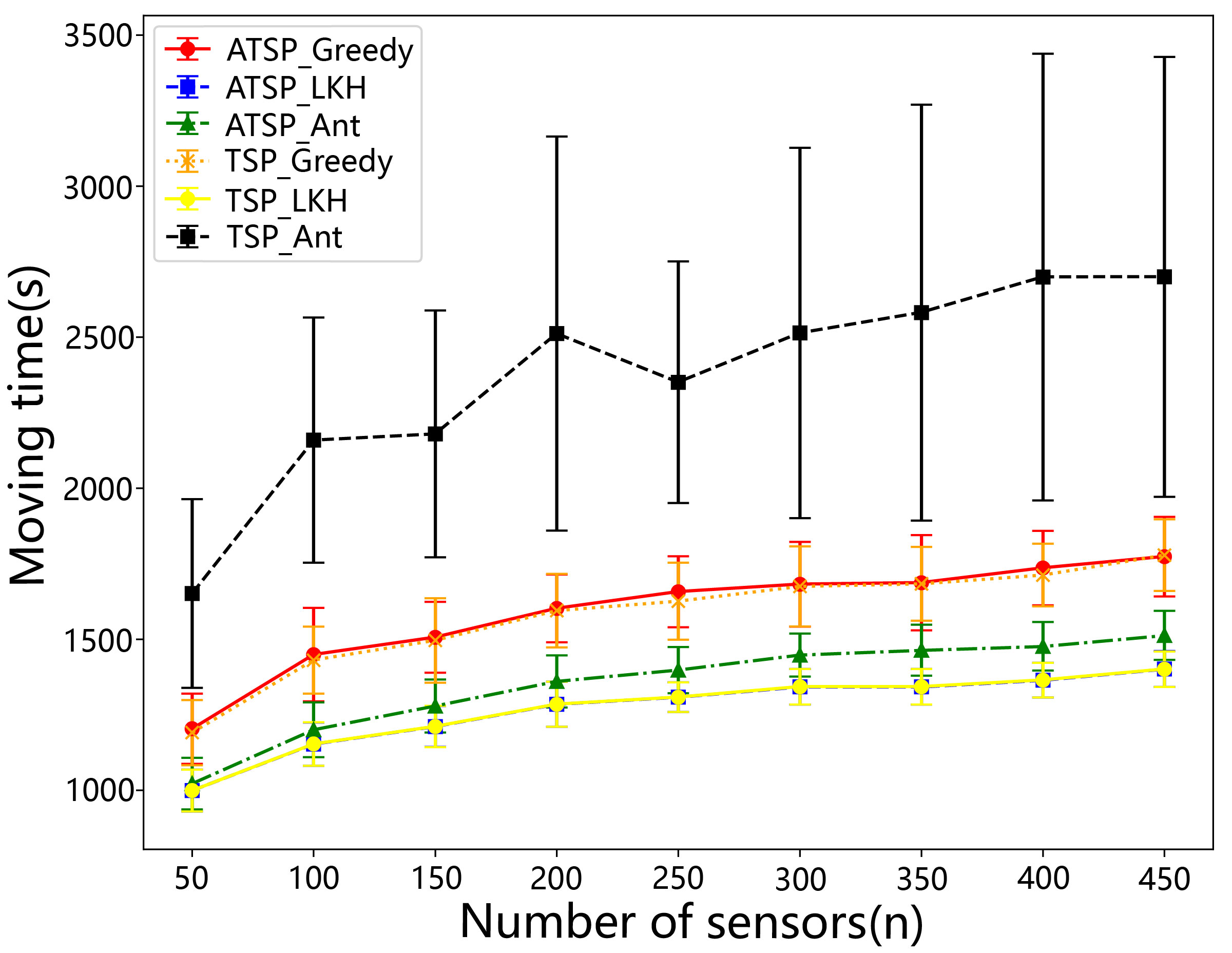}}
\label{fig8(b)}
\subfigure[Runtime]{\includegraphics[width=0.328\textwidth]{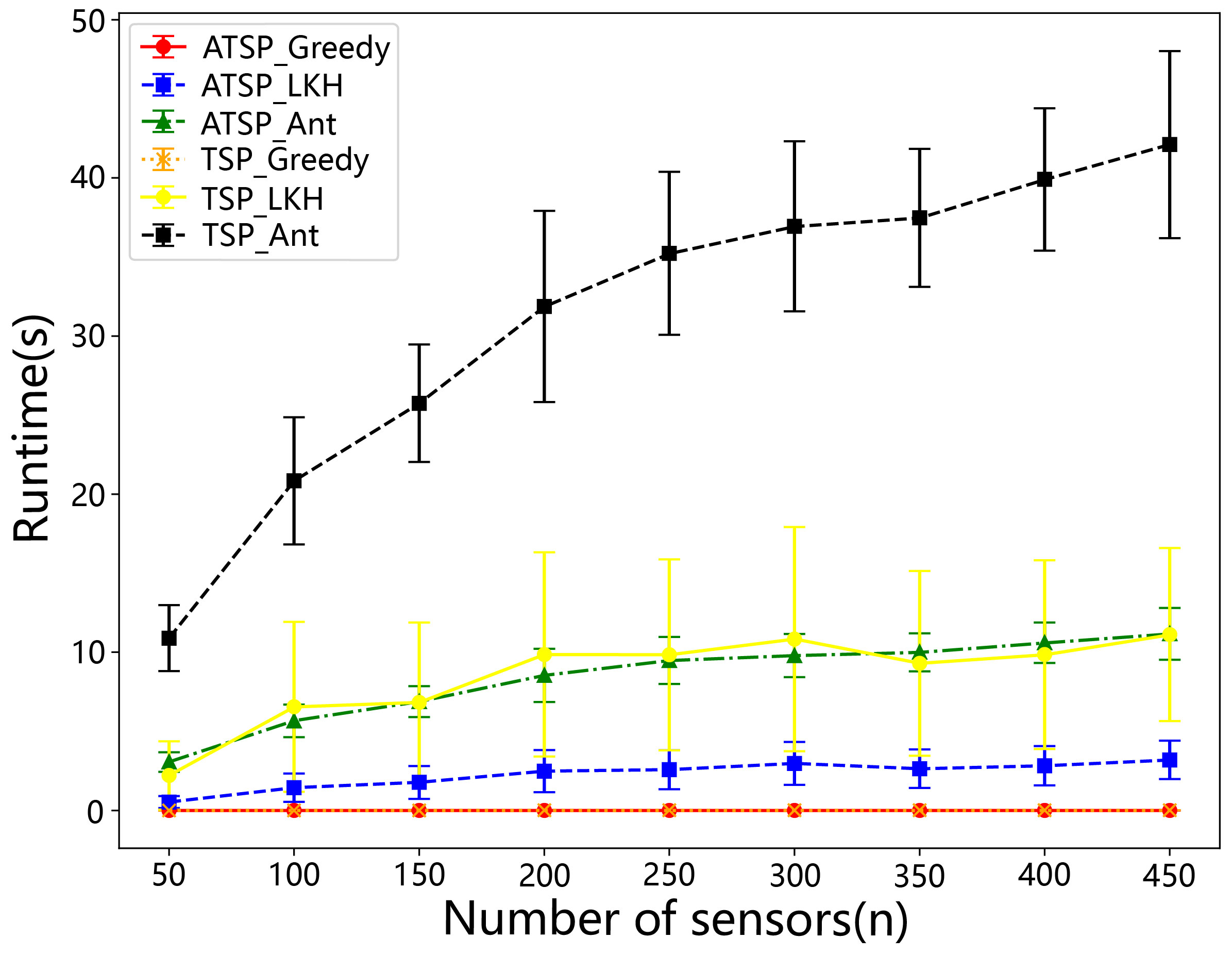}}
\label{fig8(c)}
\caption{Comparison of different approaches to solve ATSP}
\label{fig8}
\end{figure*}

Fig. \ref{fig8} shows that ATSP\_LKH demonstrates optimal overall performance in term of movement energy consumption. Although solving ATSP by converting it to an TSP problem is feasible, performance of algorithms using this approach are consistently inferior to the comparators of directly solving ATSP. Fig.~\ref{fig8} also shows that, while ATSP\_Ant performs well in solving ATSP, its comparator TCP\_Ant exhibits much poorer performance. This is because that, converting ATSP to TSP doubles the points to be visited, making the ant colony algorithm more likely to fall into local optima.

The results in Fig. \ref{fig8}(b) show that, although ATSP\_LKH performs much better than TSP\_LKH in term of movement energy consumption, as they differ only in the form of input data, they have similar performance in term of moving time.

Fig.~\ref{fig8}(c) reveals that ATSP\_Greedy has the shortest runtime, but at the expense of slight performance degradation. ATSP\_LKH ranks second in this metric. Overall, ATSP\_LKH is the most preferable one.

\section{Conclusion}
\label{sec:conclusion}

In this paper, we investigates the DMC scheduling problem for node recharging in routing-asymmetric WRSNs, termed as ADMCCS, aiming to determine the minimal energy loss charging strategy. We first prove that ADMCCS is an NP-hard problem, and subsequently decompose it into four key components: charging position generation, determination of charging directions, optimal energy transmission time length and asymmetric path planning. Regarding charging position generation, we prove its NP-hard problem and propose an adaptive KCPG algorithm to minimize the number of charging positions while covering all nodes. For charging direction determination, we employ the cMFRDS algorithm to select a minimal functional representative direction set for each position. In third part, we formulate a linear programming optimization problem to minimize the total transmission time, solved using Cplex, to determine the optimal charging times for each direction. In the asymmetric path planning phase, we utilize the LKH algorithm to find paths with minimal energy consumption based on the charging position set. Finally, we present the RA-DMCS algorithm as a comprehensive solution to the ADMCCS problem, integrating the methodologies from the four aforementioned parts. Currently, our approach still has room for improvement, and we plan to further expand our work based on ADMCCS. We did not consider issues include the presence of obstacles in the network, energy allocation among nodes, dynamic considerations of changes in node distribution, and the collaborative operation of omnidirectional and directional charging vehicles. We aim to address these problems in future work.

\bibliographystyle{IEEEtran}
\bibliography{AsymTSP}

\end{document}